\documentclass[a4paper,10pt]{article}
\usepackage[utf8x]{inputenc}
\usepackage{jheppub}
\usepackage{atlasmisc}
\usepackage{booktabs}
\usepackage{blindtext}
\usepackage{amsmath}
\usepackage{bbm}
\usepackage{placeins}
\usepackage{subfig}
\usepackage{hepunits}
\usepackage{cancel}
\usepackage{amsthm}
\usepackage{graphicx}
\newtheorem{proposition}{Proposition}[section]

\newtheorem{corollary}{Corollary}[section]
\allowdisplaybreaks

\newcommand{\mpid}{m_{\pi_D}}
\newcommand{\mx}{m_{X}}
\newcommand{\mX}{\mx}
\newcommand{\tpid}{\tau_{\pi_D}}

\begin{document}

\title{Dark showers from sneaky dark matter}
\author[a]{Adri\'an Carmona,}
\author[b]{Fatemeh Elahi,}
\author[c,d]{Christiane Scherb,}
\author[b]{and Pedro Schwaller}
\affiliation[a]{CAFPE and Departamento de F\'isica Teórica y del Cosmos,\\ Universidad de Granada, E18071 Granada, Spain}
\affiliation[b]{PRISMA$^+$ Cluster of Excellence \& Mainz Institute for Theoretical Physics,\\ Johannes Gutenberg University, 55099 Mainz, Germany}
\affiliation[c]{Berkeley Center for Theoretical Physics, Department of Physics, University of California, Berkeley, CA 94720, USA}
\affiliation[d]{Theoretical Physics Group, Lawrence Berkeley National Laboratory, Berkeley, CA 94720, USA}
\emailAdd{adrian@ugr.es}\emailAdd{felahi@uni-mainz.de}\emailAdd{cscherb@lbl.gov}\emailAdd{pedro.schwaller@uni-mainz.de}
\abstract{
We present a minimal composite dark matter model, based on a $SU(N_d)$ dark sector with $n_f$ dark quarks and a heavy t-channel mediator. For $n_f\geq 4$, the dark flavor symmetry guarantees the stability of a subset of the dark pions, which serve as our dark matter candidates. Their relic abundance is determined by co-scattering or co-annihilation with the remaining dark pions, which are unstable and decay. Due to their degenerate masses, the annihilation cross section is suppressed at low temperatures, thereby avoiding stringent constraints from indirect detection and opening up the GeV mass window. The decaying dark pions are naturally long lived. We obtain limits on the model from semi-visible or emerging jet searches and estimate the reach of future probes. 
}

\preprint{MITP/24-083}
\maketitle

\section{Introduction}

Dark matter (DM) remains one of the most significant open problems in modern particle physics. While it provides the most compelling explanation for a wide range of very precise astrophysical and cosmological observations, from the rotation curves of galaxies to the fluctuations in the cosmic microwave background, its microscopic nature remains largely unknown. In analogy with the visible sector, which features many more particles than are necessary to explain the atoms and photons we interact with daily, it is reasonable to assume that DM is part of a larger set of particles, referred to as a dark sector. 

Realistic models should at least explain why DM is stable, how it achieves the observed abundance, and why it has not yet been detected. In this work, we present a DM model based on a dark sector with a new confining gauge force. DM stability is ensured by an unbroken flavor symmetry in the dark sector, which makes some of the lightest dark sector states, the dark pions, stable. The relic abundance of DM is determined by annihilation into unstable dark pions (termed \emph{transient} dark pions) of the same mass, providing a simple realization of the impeded DM paradigm~\cite{Kopp:2016yji}. This mechanism introduces a velocity dependence to the annihilation rate, suppressing late-time annihilations and thereby \textit{sneakily} circumventing stringent constraints from indirect detection. Consequently, the model is viable for DM masses in the GeV range, which are typically excluded for thermal production mechanisms. In fact after taking into account constraints from direct detection, the preferred regions of parameter space fall either between $1-10$~GeV or above 200~GeV, with transient dark pion lifetimes required to be macroscopic ($c\tau > 1$~mm) to avoid overly strong bounds from direct detection.

Composite DM arising from pion-like states in confining dark sectors has been extensively studied (e.g.,~\cite{Essig:2009nc,Bhattacharya:2013kma,Cline:2013zca,Hochberg:2014kqa,Harigaya:2016rwr,Kopp:2016yji,Berlin:2018pwi,Beauchesne:2018myj,Beauchesne:2019ato,Bernreuther:2019pfb,Contino:2020god,Chu:2024rrv,Garcia-Cely:2024ivo})\footnote{See also~\cite{Maleknejad:2022gyf,Alexander:2023wgk,Alexander:2024nvi} for work on ultralight dark pions as DM candidates.}. For instance, \cite{Beauchesne:2018myj} examines a similar setup but focuses on the case of $n_f=3$ dark quark copies, such that several couplings have to be set to zero to ensure DM stability. In contrast, we demonstrate that for $n_f \geq 4$, DM stability follows naturally from a dark flavor symmetry, which also protects the DM candidates from radiatively induced decays. Moreover, this symmetry enforces degenerate dark pion masses, which is crucial to evade stringent indirect detection constraints in the low-mass regime. 

Compared to other models, our framework has offers several attractive features. The relic abundance primarily depends on the parameters of the confining dark sector and is independent of the mass and couplings of the heavy mediator connecting the dark and visible sectors. Since DM is among the lightest states in the dark sector, it is copiously produced in dark showers, making it directly accessible in collider experiments seeking such signatures. We find that the viable parameter space of the model is probed by a combination of collider, direct detection and flavor experiments. 

In the $1-10$~GeV DM mass range, the model exhibits a rich collider phenomenology, as the dark showers contain a large number of both long lived and stable dark pions. Consequently, its signatures are a combination of emerging jets~\cite{Schwaller:2015gea,Mies:2020mzw,Linthorne:2021oiz} and semi-visible jets~\cite{Cohen:2015toa,Cohen:2017pzm}, which can be effectively targeted by merging these search strategies. We perform a recast of the semi-visible and emerging jets searches, as well as conventional multi-jet and missing energy searches, which cover the limits of vanishing and infinite lifetime, respectively. Since DM is now efficiently produced in the dark shower, the missing energy search is also sensitive for promptly decaying dark pions. Constraints are obtained in the DM-mediator mass plane to facilitate comparison with direct detection, and benchmarks for future experimental studies are proposed. 

This paper is organized as shown in the table of contents. Have fun reading it.

\section{The Model}
The proposed model comprises a strongly interacting QCD-like dark sector that interacts with the Standard Model (SM) through a scalar bi-fundamental. Specifically, in addition to the SM, we consider a $SU(N_d)$ gauge sector and $n_f$ Dirac fermions  $Q_\alpha,\, \alpha=1,\ldots,n_f$, transforming in the fundamental representation of $SU(N_d)$.  The Lagrangian governing the dynamics of the dark sector is thus far
\begin{equation}
\mathcal{L}_D = - \frac{1}{4} (G^{\mu \nu,a}_D)^2+ 
\bar Q_\alpha  i \displaystyle{\not}{D} Q_\alpha - m_{Q,\alpha\beta} \bar Q_\alpha  Q_\beta \,,
\end{equation}
where $G_D^{\mu\nu,a}$ is the dark gluon field tensor, transforming in the adjoint of $SU(N_d)$, with $a\in\{1,\ldots,N_d^2-1\}$, and $\alpha, \beta\in \{1,\ldots,n_f\}$ are dark flavor indices.  As mentioned before, we introduce also a scalar $X$ that is a fundamental of both $SU(N_d)$ and $SU(3)_c$. This scalar particle acts as the portal between the two strongly interacting sectors by enabling the interaction
\begin{equation}
\mathcal{L}_{\rm portal} = - \kappa_{\alpha i} \bar \psi_i Q_\alpha  X + \mathrm{h.c.},
\label{eq:quarkLag}
\end{equation}
where $i$ indicates SM flavor.  The quantum numbers of $X$ under SM gauge symmetries determine the form of the operator. If $X$ is a singlet of $SU(2)_L$ and has a hypercharge of $1/3$, then $\psi_i = d_R^i$~\cite{Renner:2018fhh}. If $X$ has a hypercharge of $-2/3$ while still being a singlet of $SU(2)_L$, then this portal is with right-handed up-type quarks~\cite{Carmona:2021seb,Carmona:2022jid}, $\psi_i=u_R^i$. By making $X$ a doublet of $SU(2)_L$ with hypercharge $-1/6$, $\psi$s are the left-handed quarks in the SM, $\psi_i=\mathcal{Q}_L^i$. In this paper, we focus on the evolution of dark sector, the specific form of the portal less critical for most of the discussion. However, in sections
\ \ref{sec:DM} and\ \ref{sec:collider}  we explore how different quantum numbers of $X$ affect the phenomenology of this model. 

The dark sector has a $SU(n_f)_L\times SU(n_f)_R \times U(1)$ global symmetry. A crucial assumption is that $m_{Q\alpha\beta} = m \delta_{\alpha \beta}$, i.e. that the diagonal $SU(n_f)_V\times U(1)$ subgroup that remains after chiral symmetry breaking is not further broken by the dark quark masses. This symmetry simplifies the portal interaction. 
Using singular value decomposition, the coupling matrix $\kappa$ can be rewritten as 
\begin{align}\kappa = VDU\,,
\label{eqn:vdu}
\end{align} 
where $U$ is a $3\times 3$ unitary matrix, $D$ is a $n_f \times 3$ diagonal matrix, and $V$ is a $n_f \times n_f$ unitary matrix. 
By performing a $SU(n_f)_V$ rotation, $V$ can be eliminated from Eq.~\eqref{eqn:vdu}, leaving only the first three dark quarks coupled to the SM ($\kappa_{\alpha i} = 0$ for $\alpha > 3$). Consequently, the coupling Eq.\ \eqref{eq:quarkLag} preserves a $SU(n_f-3)\times U(1)$ subgroup of the dark flavour symmetry. This symmetry ensures the stability of a subset of the dark pions, which become our DM candidates. 
Let us denote the dark quarks that don't interact with the SM by $Q_{\alpha_D}$, and the coupling associated with them as $\kappa_{\alpha_D i}$. Loop level corrections to $\kappa_{\alpha_D i}$ are guaranteed to be zero as well, due to the unbroken global symmetry.

Let us continue with our analysis of the symmetry structure.  
In the limit where $m_{Q}\ll \Lambda_D$, the $SU(n_f)_L \times SU(n_f)_R$ chiral symmetry is only softly broken. Similar to SM QCD, we assume $SU(N_d)$ undergoes confinement at the scale $\Lambda_D$, and the theory is described by dark hadrons at low energies. Given that after confinement, we have a non-vanishing expectation value $\langle \bar Q_L Q_R \rangle $ in the vacuum, the chiral symmetry is spontaneously broken to $SU(n_f)_V$. Because of spontaneous symmetry breaking, we have $n_f^2 -1$ pseudo Nambu Goldstone Bosons (pNGB) with a mass proportional to the dark quark masses. If we do not have any additional breaking of the dark flavor symmetry ($m_Q \equiv m_{Q,\alpha \beta}$ for all $\alpha$ and $\beta$), all pions have the same mass at tree-level.~\footnote{At the loop level, a small mass splitting will be generated between the stable and transient dark pions.} At low energies, the dynamics of these dark pions determine the cosmological evolution of the dark sector. The Lagrangian describing the self-interaction of dark pions is given by
\begin{align}
    \mathcal{L}_{\rm dChPT}=\frac{f_D^2}{4}\mathrm{Tr}\left(\partial_{\mu}U_D \partial^{\mu}U^{\dagger}_D\right)+\frac{f_D^2 B_{D}}{2}m_{Q}\mathrm{Tr}\left(U_D^{\dagger}+U_D\right),
    \label{eq:self}
\end{align}
where $B_D$ is a constant related to the dark pion mass, $m_{\pi_D}^2=2 B_D m_Q$, whereas  the portal becomes
\begin{equation}
\mathcal{L}^{\rm portal}_{\rm dChPT}= i \frac{f_D^2}{4 m_X^2} \kappa_{\alpha i} \kappa_{\beta j}^{\ast}\left\{ \text{Tr}(c_{\beta\alpha} U^\dagger_D \partial_\mu U_D)(\bar \psi_i \gamma^\mu P_R \psi_j)+\text{Tr}(c_{\beta\alpha} U_D \partial_\mu U_D^{\dagger})(\bar \psi_i \gamma^\mu P_L \psi_j)\right\}\,,
\label{eq:pionLag}
\end{equation}
where
\begin{equation}
U_D = \text{Exp}\left[ \frac{2i}{f_D} \Pi_D\right]\,,
\end{equation}
with $\Pi_D = \pi_D^a T^a$ while $T^a$ being the generators of $SU(n_f)$, normalized as $\mathrm{Tr}(T^aT^b)=\delta^{ab}/2$. For instance, in the $n_f=3$ case,  $T^a=\lambda^a/2$ with $\lambda^a$ the Gell-Mann matrices. The decay constant of dark pions in Eqs.\ \eqref{eq:self} and \eqref{eq:pionLag} is denoted by $f_D$. The matrix $c_{\alpha\beta}$ is defined as
\begin{align}
(c_{\alpha\beta})_{\rho \lambda}=\delta_\alpha^{\rho}\delta_\beta^{\lambda}, \qquad \alpha,\beta,\rho,\lambda\in\{1,\ldots,n_f\}.
\end{align}
As discussed earlier, if $n_f>3$, some dark quarks $Q_{\alpha_D}$ have no tree-level interaction with the SM sector. 
Hence, dark pions containing $Q_{\alpha D}$ quarks and with non-trivial charges under $\mathcal{G}_{\rm DM}\equiv SU(n_f-3)\times U(1)\subset SU(n_f)_V$ will be stable. 
Taking into account the decomposition of the adjoint of $SU(n_f)_V$ with respect to its maximal subgroup $SU(3)\times SU(n_f-3)\times U(1)\subset SU(n_f)_V$, one can see that the stable dark pions transform as (see Eq.\,\eqref{eq:adj})
\begin{align}
\varphi\sim (\mathbf{3},\overline{\mathbf{n_f-3}})_{\frac{n_f a}{n_f-3}}+\mathrm{h.c},\quad \phi\sim (\mathbf{1},\mathbf{(n_f-3)^2-1})_0,\quad |a|=\sqrt{\frac{n_f-3}{6n_f}}.
\end{align}
The number of stable dark pions is therefore $2\times 3\times (n_f-3)+(n_f-3)^2-1=n_f^2-10$.  In general, we represent the stable pions by $\pi_{\rm DM}$, and the rest of dark pions by $\pi_{\rm tran}$. The $(\varphi, \phi)$ notation is only used if we have to distinguish the different representations under the DM symmetry, as in the discussion of direct detection. 
In the particular case of $n_f=4$, we have six $\pi_{\rm DM}$ (transforming as $\mathbf{3}_{\sqrt{\frac{2}{3}}}$ under $SU(3)\times U(1)$) and nine $\pi_{\rm tran}$.

Given the absence of tree-level interactions between $\pi_{\rm DM}$ and the SM sector (Eq.\ \eqref{eq:pionLag}), the self-interactions of dark pions determine the evolution of $\pi_{\rm DM}$ in the early universe
\begin{align}
\mathcal{L}_{\rm SI} =& - \frac{2}{3 f_D^2}\text{Tr}\left(\Pi_D^2  \partial_\mu \Pi_D \partial^\mu \Pi_D- \Pi_D \partial_\mu \Pi_D \Pi_D \partial^\mu \Pi_D \right)\nonumber \\
&+\frac{2 N_d}{15 \pi^2 f_D^5} \epsilon^{\mu\nu\rho\sigma} \text{Tr}\left(\Pi_D \partial_\mu \Pi_D \partial_\nu \Pi_D \partial_\rho \Pi_D \partial_\sigma \Pi_D\right)+ O(\Pi_D^6)\,. 
\label{eq:pionSI}
\end{align}
The first term in Eq.\ \eqref{eq:pionSI} comes from Eq.~\eqref{eq:self} and describes a four-point interaction between dark pions. The second line in Eq.\ \eqref{eq:pionSI} describes a five-point $\pi_D$ interaction that is known as the Wess-Zumino-Witten term~\cite{Wess:1971yu,Witten:1983tw}. Using some $\mathfrak{su}(n_f)$ algebra,  these terms can be rewritten as   
\begin{align}
\mathcal{L}_{\rm SI}^{(4)} =& - \frac{2}{3 f_D^2}\pi_D^a \pi_D^b\partial_{\mu}\pi_D^c\partial^{\mu}\pi^d_D\Bigg[\frac{1}{4}f^{acm}f^{bdm}\Bigg]\,,\\
\mathcal{L}_{\rm SI}^{(5)}=&\frac{2 N_d}{15\pi^2f_D^5}\epsilon^{\mu\nu\rho\sigma}\pi^a\partial_{\mu}\pi^b\partial_{\nu}\pi^c\partial_{\rho}\pi^d\partial_{\sigma}\pi^e\Bigg[-\frac{1}{16}\big(f^{abf}f^{cdg}d^{efg}\big)\Bigg],
\label{eq:pionSIexp}
\end{align}
where $f^{abc}$ are the $\mathfrak{su}(n_f)$ structure constants and $d_{abc}$ is a totally symmetric third rank tensor, defined in Appendix~\ref{app:sun}. By integrating by parts, the WZW term can be rewritten as follows
\begin{align}
    \mathcal{L}_{\rm SI}^{(5)}=&\frac{2 N_d}{15\pi^2f_D^5}\epsilon^{\mu\nu\rho\sigma}\sum_{a<b<c<d<e}\pi_D^a\partial_{\mu}\pi_D^b\partial_{\nu}\pi_D^c\partial_{\rho}\pi_D^d\partial_{\sigma}\pi_D^eT_{abcde},\\
  T_{abcde}&=\sum_{\pi \in S_5}\mathrm{sign}(\pi)\Delta_{\pi\big(\{a,b,c,d,e\}\big)},\qquad  \Delta_{abcde}=-\frac{1}{16}\Big(f^{abf}f^{cdg}d^{efg}\Big).
\end{align}
Among other terms, the mass term leads to a four-point contact interaction among the dark pions
\begin{align}
\mathcal{L}_{\rm m}=\frac{f_D^2 B_D}{2}m_Q \mathrm{Tr}\Big(U_D^{\dagger}+U_D\big)=-\frac{m_{\pi_D}^2}{2}\pi_D^a \pi_D^a+\frac{m_{\pi_D}^2}{6f_D^2}c_{abcd}\pi_D^a\pi_D^b\pi_D^c\pi_D^e+\mathcal{O}(1/f_D^4)
\end{align}
where
\begin{align}
c_{abcd}=\frac{1}{n_f}\delta^{ab}\delta^{cd}+\frac{1}{2}\big(d^{abm}d^{cdm}-f^{abm}f^{cdm}\big).
\end{align}
These interactions define four-point and five-point interactions among dark pions. In the particular case of $n_f=4$, on which we will focus throughout the rest of the paper, there is an accidental symmetry making $\pi_{\rm DM}$ to appear always in pairs in these self-interactions (see Appendix~\ref{ap:sym}). In general, these interactions produce an indirect portal of  $\pi_{\rm DM}$ with the SM sector by interacting with $\pi_{\rm tran}$. The cross sections that determine the relic abundance therefore only depend on $N_d, n_f, f_D $, and $m_{\pi_D}$, while the rest of the UV parameters can be ignored for now. 

The decay of $\pi_{\rm tran} $ to SM particles is given by
\begin{align}
	\Gamma(\pi_{\rm tran}^{(\alpha,\beta)}\to q_i \bar{q}_j)&=\frac{N_c f_D^2 m_{\pi_D}}{128\pi m_{X}^4}\left| \kappa_{\alpha i}\kappa_{\beta j}^{\ast}\right|^2  (m_{q_i}^2+m_{q_j}^2)\mathfrak{Q}_{ij} \\
	\Gamma(\pi_{\rm tran}^b\to q_i \bar{q}_j)&=\frac{N_c f_D^2 m_{\pi_D}}{64\pi m_{X}^4}\left|\sum_{\alpha,\beta=1}^{4} \kappa_{\alpha i}\kappa_{\beta j}^{\ast}(T^b)_{\alpha \beta}\right|^2(m_{q_i}^2+m_{q_j}^2)\mathfrak{Q}_{ij} ,
	\label{eq:partdec}
\end{align}
where $\mathfrak{Q}_{ij}$ is the kinematic factor 
\begin{align}
\mathfrak{Q}_{ij}=\left[1-\frac{\left(m_{q_i}^2-m_{q_j}^2\right)^2}{(m_{q_i}^2+m_{q_j}^2)m_{\pi_D}^2}\right] \sqrt{\left(1-\frac{\left(m_{q_i}+m_{q_j}\right)^2}{m_{\pi_D}^2}\right)\left(1-\frac{\left(m_{q_i}-m_{q_j}\right)^2}{m_{\pi_D}^2}\right)}
\end{align}
and $\pi_{\rm tran}^{(\alpha,\beta)}$, with $\alpha < \beta\in\{1,2,3\}$, are the  off-diagonal pions and $\pi_{\rm tran}^{b}$, with  $b = 3, 8,15$, the diagonal ones, with $T^a$ the generators of $\mathfrak{su}(4)$ defined in Appendix~\ref{app:sun}.  We can see that these decay widths are proportional to $\kappa_{\alpha i} \kappa_{\beta j}^{\ast}/m_X^2$, and thus depends on the exact structure of the $\kappa$ matrix and $m_X$. Barring specific cancellations, we can however assume that all decay widths are of the same order, and we therefore use a common decay rate $\Gamma_{\rm tran}$ when discussing the cosmological evolution of the dark sector.

\section{Sneaky dark matter}
\label{sec:DM}
\subsection{Relic abundance}

Similar to the SM QCD sector, our proposed dark sector consists of many dark baryons and mesons. A thorough calculation of the evolution of these particles in the early universe is well beyond the scope of our paper. 
Fortunately all heavy states annihilate efficiently or decay into dark pions and thus can be neglected in practice in the computation of relic abundance\footnote{This is in general satisfied as long as $m_{\pi_D} \ll 4 \pi f_D$, i.e. as long as the chiral EFT is perturbative. Instead all other interaction rates are non-perturbative and one can expect them to saturate the unitarity bound for DM annihilation, which suggests that the relic abundance of the heavier states is negligible as long as the masses of potentially stable states (in particular, the lightest dark baryons), is below $100$~TeV.}. Therefore we can focus on the evolution of dark pions.

As mentioned above, for the particular case of $n_f=4$, thermally averaged cross-sections involving an odd number of stable dark pions will vanish. Therefore, the relevant co-annihilation cross-section for this case is that of two stable dark pions going into two transient dark pions
\begin{align}
\langle \sigma v\rangle_{2_{\rm DM}\to 2_{\rm tran}}\stackrel{n_f=4}{=}\frac{m_{\pi_D}^2}{\pi^{3/2} f_D^4 \sqrt{x}}\left[\frac{1171}{576}\right]\approx\frac{2m_{\pi_D}^2}{\pi^{3/2} f_D^4 \sqrt{x}}\approx \sigma_0 v,\qquad \sigma_0=\frac{2m_{\pi_D}^2}{\sqrt{3}\pi^{3/2}f_D^4},
\end{align}
where $x=m_{\pi_D}/T$, see Appendix \ref{sec:ap_relic} for a detailed derivation. A notable feature here is the velocity suppression of the annihilation cross section at low temperatures, $v\propto 1/\sqrt{x}$, which arises due to the degenerate dark pion masses. This has a modest effect on the relic abundance, but is crucial for avoiding
stringent constraints from indirect detection and from late energy injection into the cosmic microwave background (see also~\cite{Kopp:2016yji}). As mentioned previously, a small mass splitting  $\Delta$ between the stable and transient dark pions is generated at the loop level. This will not change the results provided that $\Delta/m_{\pi_D}\ll 1$, which is expected in our case. 

The WZW term induces five-point interactions among dark pions which give rise to $3\to 2$ processes. They scale as
\begin{align}
\langle \sigma v^2\rangle_{3\to 2}=\frac{25\, k\,  m_{\pi}^5 N_d^2}{32\sqrt{5}\pi^5 f_D^{10} N_{\rm DM}^3 x^2}\,,
\end{align}
where  $k=\mathcal{O}(1)$ is a combinatorial factor and $N_{\rm DM}=6$ (the reader can find explicit expressions in Appendix~\ref{sec:ap_relic}).
The full Boltzmann equations (BEs) governing the evolution of dark matter pion ($n_{\rm DM}$) and transient dark pions ($n_{\rm tran}$) are presented in Appendix~\ref{sec:ap_relic}. As one expects, these BEs are highly coupled and include many terms that may effect the evolution of number densities. However, let us make the following remarks which simplify the BEs significantly. 
\begin{itemize}
 \item[-] We have explicitly checked that $2\to 2$ processes are more efficient than the $3\to 2$ ones, for all values of $m_{\pi_D}$, $f_{D}$ and $x$ under consideration. Therefore, we will in the following safely ignore the contribution of $3\to 2$ processes.

 \item[-] To ensure kinetic equilibrium between the SM and dark sector, we require the decay widths of transient dark pions to be larger than the Hubble rate: $\Gamma(\pi_{\rm tran}) \gg H $, with the Hubble rate being defined as $ H = 1.66 \sqrt{g_\star} T^2/M_{\rm Pl}$, where $g_\star$ is the effective degree of freedom, and $M_{\rm Pl}$ being the Planck mass.
 To be conservative we enforce it for $T= m_{\pi_D}$, and we get $c\tau (m_{\pi_D}/\mathrm{GeV})^2 \ll c \hbar M_{\rm Pl}/(1.66 \sqrt{g_\star})\simeq 100 \,\mathrm{m}$,   where $c\tau$ is the lifetime of the transient dark pion. The region to the right of the purple lines in Fig. \ref{fig:relicDD} violates this condition for $c\tau = 10$ cm and $c\tau=1$ mm. This region is not necessarily excluded, but rather requires a more careful analysis of the relic abundance.  

 \item[-] Due to the fast decay of $\pi_{\rm tran
    }$ compared to the Hubble rate, for $T < m_{\pi_D}$ the abundance of transient dark pions is Boltzmann suppressed. Therefore, the relic abundance of dark pions is solely due to annihilation of dark pions to transient dark pions while the back reactions are negligible.
\end{itemize}
Given the assumptions above, the evolution of $n_{\rm DM}$ simplifies to 
\begin{align}
&\dot{n}_{\rm DM} + 3 H n_{\rm DM}= -\langle \sigma v\rangle_{2_{\rm DM}\to 2_{\rm tran}}\left[n_{\rm DM}^2-(n_{\rm DM}^2)_{\rm eq} \right],
\label{eq:BEsimplified}
\end{align}
where $(n_{\rm DM})_{\rm eq} =\frac{ Tm_{\pi_D}^2}{2\pi^2}  K_2 \left(\frac{m_{\pi_D}}{T}\right)$ is the equilibrium number density. It is more convenient to convert $n_{\rm DM}$ to dimensionless variable $Y_{\rm DM} = n_{\rm DM}/s$, where $s = \frac{2 \pi^2}{45} g_{\star} T^3$ is the entropy density. Assuming adiabatic expansion (i.e., $ s a^3$ is constant), the BE describing the evolution of $\pi_{\rm DM}$ abundance is
\begin{equation}
    \frac{dY_{\rm DM}}{dx} =- \frac{\langle \sigma v\rangle s }{x H} \left( 
Y_{\rm DM}^2 - (Y_{\rm DM}^2)_{\rm eq} \right),
\end{equation}
where $(Y_{\rm DM})_{\rm eq} = (n_{\rm DM})_{\rm eq}/s$. The simple form of $\sigma_0$ results in a straightforward relation between $m_{\pi_D}$ and $f_D/m_{\pi_D}$ for which the observed relic abundance is obtained. It is shown in Fig.~\ref{fig:relicDD} together with other constraints. The solid red line corresponds to dark pions constituting all of DM, and the dashed line is when they constitute 10\% of the total DM in the universe. The gray region is excluded due to overproduction of DM.

\subsection{Direct detection}
\label{sec:DD}

The expected rate of DM scattering off of a nucleus in  direct detection experiments is:
\begin{equation}
    \frac{dR}{dE_R} = \frac{\rho_{\rm DM} M_{\rm det}}{m_A m_{\rm DM}}\int_{v_{\rm min}}^{v_{\rm esc}}{v f(v) \frac{d\sigma}{d E_R} dv},
\end{equation}
where $\rho_{\rm DM}$ is the local DM density, $M_{\rm det}$ is the mass of the detector, $m_A$ and $ m_{\rm DM}$ are respectively the nucleus and DM mass, $f(v)$ is the normalized DM velocity distribution, $\sigma$ is the scattering cross section of DM with the nucleus, and  $E_R$ is the recoil energy. The bounds of the integral $v_{\rm min}$ and $v_{\rm esc}$ correspond to the minimal velocity required to induce a recoil of energy $E_R$, and the maximum velocity to be bound to the potential well of Milky Way, respectively. 
The DM-nucleus interaction is given by~\cite{Schumann:2019eaa}
\begin{equation}
    \frac{d\sigma}{d E_R} = \frac{m_A}{2 v^2 \mu^2} \left( \sigma_{\rm SI} F^2_{\rm SI} (E_R) + \sigma_{\rm SD} F^2_{\rm SD} (E_R) \right),
\end{equation}
where $\mu= m_{\rm DM} m_A/ (m_{\rm DM} + m_A)$ is the reduced mass of DM-nucleus system. The classification of DM cross section with nucleons as spin-independent (SI) or spin-dependent (SD) is determined by the nature of the DM-nucleon interaction.\footnote{Generally, the inclusion of $ \gamma^5 $ (axial current) indicates an SD cross section.} The form factors $ F_{\rm SI} $ and $ F_{\rm SD} $ characterize the coherence of DM interaction with nuclei\footnote{At small momentum transfer , contributions from all nucleon partial waves combine, resulting in coherent scattering of DM with the entire nucleus. As the momentum transfer increases, the de Broglie wavelength of the DM decreases, leading to interactions that involve only a portion of the nucleus. This loss of coherence is accounted for by the form factors $ F_{\rm SI} $ and $ F_{\rm SD} $.} and can be calculated based on the properties of the target nuclei~\cite{Schumann:2019eaa, Helm:1956zz}. Nonetheless, direct detection experiments present the limit as a function of $\sigma_{\rm SI}$ and  $\sigma_{\rm SD}$, alleviating the need for phenomenologists to address these form factors.

In our model, the scattering of dark pions with cold nucleons in direct detection experiments is governed by Eq.\ \eqref{eq:pionLag}.  Expanding $U_D^{\dagger} \partial_{\mu} U_D$ and $U_D \partial_{\mu} U_D^{\dagger}$ to the second order in $\pi_D$ yields in both cases $2/f_D^2 \big[ \Pi_D,\partial_{\mu}\Pi_D\big]$ and thus
\begin{align}
\mathcal{L}_{\rm dChPT}^{\rm portal} \supset  
&-\frac{1}{2m_X^2} \kappa_{\alpha i} \kappa_{\beta j}^{\ast}(T^c)_{\alpha\beta}f^{abc}\pi_D^a \partial_{\mu}\pi_D^b (\bar \psi_i \gamma^\mu \psi_j),
\label{eq:pionscattering}
\end{align}
where we have used that $\mathrm{Tr}(c_{\beta\alpha} T^a)=(T^a)_{\alpha\beta}$ and the $\mathfrak{su}(n_f)$ algebra. Henceforth, we will work in the dark flavor basis where $\kappa_{\alpha i}=0$, for $\alpha>3$. This corresponds to $\kappa=D\cdot U$ with $D$ a diagonal $4\times 3$ matrix and $U$ a $3\times 3$ unitary matrix. Recall that depending on the quantum numbers of $X$, $\psi_i$ can be any of $\mathcal{Q}_L^i, u_R^i$, or $d_R^i$. %\footnote{Since the Parton Distribution Functions (PDFs) of 2nd and 3rd generation quarks in direct detection experiments is minimal, we narrow the definition of $q$ to only up and down quarks.}  
 Hence, dark pions interact with both
$$(\bar q \gamma^\mu q) \hspace{0.5 in} \text{and} \hspace{0.5 in}  (\bar q \gamma^\mu \gamma^5 q),\qquad q=(u,d,s)^T.$$
For the model at hand, we know that the stable dark pions form a (complex) triplet of $SU(3)\subset SU(4)_V$, that we can call $\varphi$. Then, the equation above leads to
\begin{align}
\mathcal{L}_{\rm dChPT}^{\rm portal} \supset  
&-\frac{1}{4m_X^2} \kappa_{m i} \kappa_{\ell j}^{\ast} (\varphi^{\ast}_\ell i
\overset{\leftrightarrow}{\partial_{\mu} }\varphi_m )\bigg(\bar q \gamma^\mu \bigg[c_{Rij}^{\psi}P_R+c_{Lij}^{\psi}P_L\bigg]q \bigg),
\end{align}
where $m,\ell \in\{1,2,3\}$ and 
\begin{align}
(c_{Rij}^{u_R})_{m\ell}&=\left\{\begin{array}{ll} \delta_{i}^m\delta_j^\ell&i,j=1\\
0&\mathrm{otherwise}\end{array}\right.,\quad (c_{Rij}^{d_R})_{m\ell}=\left\{\begin{array}{ll} \delta_{i+1}^m\delta_{j+1}^\ell&i,j=1,2\\
0&\mathrm{otherwise}\end{array}\right.,\quad (c_{Rij}^{\mathcal{Q}_L})_{m\ell}=0,\\
(c_{Lij}^{u_R})_{m\ell}&=0,\quad(c_{Lij}^{d_R})_{m\ell}=0,\quad (c_{Lij}^{\mathcal{Q}_L})_{m\ell}=\left\{\begin{array}{ll} \delta_{i+1}^m\delta_{j+1}^\ell+\delta_i^m\delta_i^\ell&i,j=1,\\\delta_{i+1}^m\delta_{j+1}^\ell & i=2\lor j=2\\
0&\mathrm{otherwise}\end{array}\right. .
\end{align}

If we assume that the velocity of the incoming $\pi_D$ is negligible compared to the masses (i.e.,
$s = (m_{\pi_D} + m_N )^2$),
we can neglect axial-vector interactions because they lead to velocity-suppressed contributions to the direct detection cross section. In this case, the relevant interactions with the nucleons $N=p,n$ read
\begin{align}
\mathcal{L}_B\supset -\frac{1}{8m_X^2}\big\{ 2 \bar{p}
\gamma^{\mu}p+\bar{n}\gamma^{\mu}n\big\} \kappa_{m1}\kappa_{\ell 1}^{\ast}(\varphi_\ell^{\ast}i\overset{\leftrightarrow}{\partial_{\mu} }\varphi_m),
\end{align}
for the $\psi_i=u_R^i$ case, and the same after exchanging $n\leftrightarrow p$ for both  $\psi_i=d_R^i$ and $\mathcal{Q}_L^i$. This interaction leads to the following cross section with the nucleons 
\begin{align}
\sigma_{N}=F_N^2 \frac{\mu^2_N}{\pi},\quad F_p^2=\sum_{m\ell}\frac{|\kappa_{\ell 1}\kappa_{m1}^{\ast}|^2}{48m_X^4}=4F_{n}^2
\end{align}
for the $\psi_i=u_R^i$ case, where $\mu_N$ is the reduced mass of DM-Nucleon system, with a similar result for the $\psi_i=d_R^i,\mathcal{Q}^i_L$ after the exchange $n\leftrightarrow p$. The SI cross section is given by~\cite{Schumann:2019eaa}:
\begin{equation}
    \sigma_{\rm SI} = \sigma_N \frac{\mu^2}{\mu_N^2} \left( (A-Z) + F_{p}/F_{n} Z\right)^2,
\end{equation}
where $A$ and $Z$ are respectively atomic and proton number.  The resulting SI cross section is 
\begin{equation}
\sigma_{\rm SI}=\sum_{m\ell} \frac{\left|\kappa_{\ell1}\kappa^{\ast}_{m1}\right|^2}{192 \pi m_X^4}\mu^2
\begin{cases} 
\left(A+Z\right)^2,&  \psi=u_R \\  \\
\left(2A-Z\right)^2,& \psi=d_R, \mathcal{Q}_L
\end{cases}.
\label{eq:DDxsectionSI}
\end{equation}
As Eq. \eqref{eq:DDxsectionSI} illustrates, $\sigma_{\rm SI}$ scales as $A^2$. Instead $\sigma_{\rm SD}$ is proportional to the spin of nucleus, which is $\mathcal{O}(1)$~\cite{Engel:1992bf}, and therefore not enhanced by the atomic number. We therefore neglect $\sigma_{\rm SD}$ and focus solely on the constraints arising from $\sigma_{\rm SI}$ in the following.

Using the current best limits on direct detection~\cite{DirectDetection,SENSEI:2023zdf}, we establish the constraints on our model, as illustrated by the green lines in Fig.~\ref{fig:relicDD}. Here the direct detection rates are rescaled with the actual abundance of $\pi_{\rm DM}$ for each point. 
 We assume that $\kappa_{ij}$ is uniform for all values of $i, j = \{1, 2, 3\}$, and we substitute $m_X$ with the lifetime of the dark pion using Eqs.~\eqref{eq:partdec}. The dashed green line represents the direct detection bound for $c\tau = 1$ mm, while the solid green line corresponds to the case where $c\tau = 10$ cm. 
 Notably, the direct detection bounds exhibit sudden fluctuations around masses on the order of $\mathcal{O}(1)$ GeV. This is primarily due to our choice of showing contours of constant $c\tau $, which lead to rapid changes in the underlying model parameters $\kappa$ and $m_X$ when the dark pion mass crosses certain thresholds. 
 Similarly, the direct detection constraints for a dark pion mass of $\mathcal{O}(100)$ GeV, with the choice of $\psi_i = d_R^i$, are stronger because the phase space for the dark pion with $\psi_i = u_R^i$ near the top mass is more constrained. The selection of $c\tau$ values is primarily motivated by the available phase space in the $\mathcal{O}(1)$ GeV dark pion mass range, which contributes more than \(10\%\) of the total dark matter and may be accessible in collider experiments.

\subsection{Indirect detection and CMB bounds}
\label{sec:ID}

One of the main probes of $\mathcal{O}(10-1000) \ \text{GeV}$ scale dark matter is through DM self annihilation to SM particles today (which in our case is through a cascade decay), leading to
diffuse gamma-ray production at the center of galaxies. The strongest constraints of such radiation comes from  \textit{Fermi}-Large Area satellite Telescope (\textit{Fermi}-LAT), which surveys a large fraction of the sky with unprecedented sensitivity~\cite{Fermi-LAT:2016afa}. The gamma-ray flux from the annihilation of DM has the form: 
\begin{equation}
\frac{\mathrm{d} \Phi_\gamma}{\mathrm{d E}_\gamma}  = \frac{1}{8 \pi m_{\pi_D}^2}
\underbrace{\sum_f \langle \sigma v \rangle_{\rm 2\pi_{DM}\to 2\pi_{\rm tran} \to 4 f }\frac{\mathrm{d N}_\gamma^f}{\mathrm{d E}_\gamma}
}_\text{DM model}
\times
\underbrace{\int\limits_{\Delta \Omega} \mathrm{d}\Omega^{\prime}
\int\limits_{\mathrm{los}} \rho^2\mathrm{d}l(r,
\theta^{\prime})}_{\text{Astrophysics part}}
\end{equation}

The left-hand side represents the measurable gamma-ray flux, while the right-hand side consists of two components: (a) a particle physics term that depends on the dark matter mass \( m_{\pi_D} \), the velocity-averaged annihilation cross-section \( \langle \sigma v \rangle \) to standard model (SM) particles (denoted as \( f \)), and the spectrum of gamma rays produced from the decay of dark matter to SM particles \( \left( \frac{\mathrm{d} N_\gamma^f}{\mathrm{d} E_\gamma} \right) \); and (b) an astrophysical term \( J(E) \) (referred to as the J-factor), which is defined by the line-of-sight integral of the square of the dark matter density \( \rho \).

In our model, the primary process affecting indirect astrophysical observations is \( \pi_{\rm DM} \pi_{\rm DM} \to \pi_{\rm tran} \pi_{\rm tran} \), followed by the subsequent decay \( \pi_{\rm tran} \to j j \). \footnote{It is important to note that for very light dark pions, the transient dark pion decays directly to photons, which would produce spectral lines. However, this region of parameter space is not addressed in this paper.} Since the indirect detection bounds are similar regardless of whether \( j \) represents up-type or down-type quarks~\cite{Hoof:2018hyn}, we do not differentiate between them.

As previously mentioned, dark matter co-annihilation into transient dark pions is velocity suppressed, specifically \( \langle \sigma v \rangle \simeq \sigma_0 v \). This results in weaker signals from objects with low dark matter velocities, such as dwarf galaxies (see also Refs.~\cite{Dror:2016rxc, Okawa:2016wrr, Jia:2016pbe}). Following Ref.~\cite{Elor:2015bho} for the cascade decay and Refs.~\cite{Hoof:2018hyn, Hutten:2022hud} for the current indirect detection bounds on the cross section of dark matter with SM quarks, we can derive the indirect detection bounds for our model. We assume the dark matter velocity in our galaxy to be \( 220 \) km/s. The constraint from indirect detection is illustrated by the dashed blue line in Fig.~\ref{fig:relicDD}. The region below the blue line corresponds to a more efficient annihilation cross section, and is excluded according to indirect detection.

 Similarly, observations of the cosmic microwave background (CMB) impose a significant restriction on any model in which DM can be annihilated. Specifically, the additional energy introduced into the early universe’s plasma due to DM annihilation could postpone the recombination process, leaving noticeable traces in the CMB. The influence of DM on the CMB is quantified by the “energy deposition yield”, represented as $$p_{\rm ann} = f_{\rm eff} \frac{\langle \sigma v \rangle}{m_{\pi_D}}.$$
Here, $f_{\rm eff}$ gives the efficiency with which the energy released in DM annihilation is absorbed by the primordial plasma, which here we assume 1 to obtain a more rigorous bound.  We demand that $p_{\rm ann }< 3.5 \times  10^{-28} \mathrm{cm}^3 \mathrm{s}^{-1} \mathrm{GeV}^{-1}$~\cite{Planck:2018vyg}. This constraint is illustrated by an orange region in Fig.  ~\ref{fig:relicDD}. As one can see the bounds from CMB are more stringent than those from indirect detection. This is primarily attributed to the higher velocity of DM particles during the era of recombination compared to their velocities today.

\subsection{Current astrophysical and cosmological constraints}

	\begin{figure}[!h]
		\begin{center}
	\includegraphics[width=0.502\textwidth]{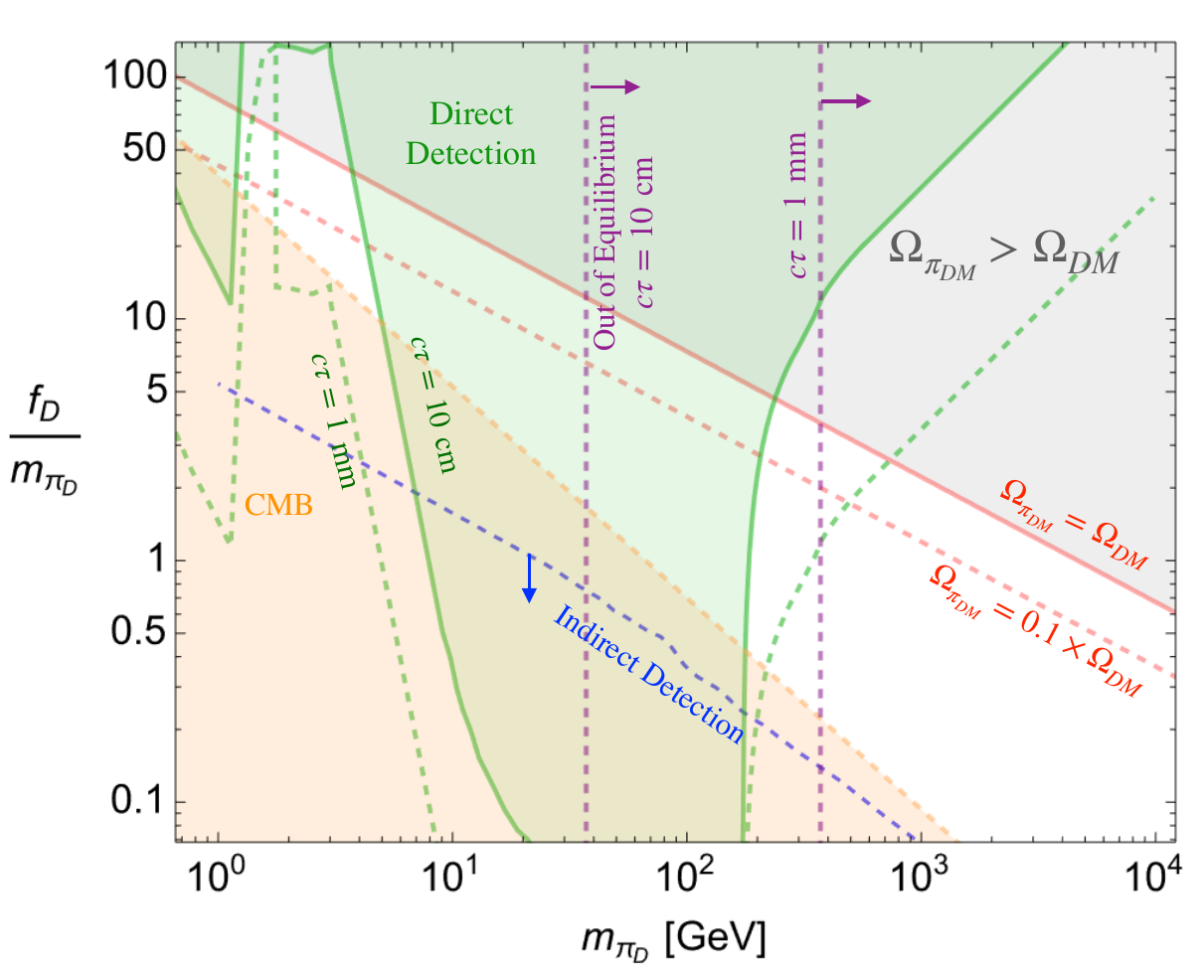}
 \includegraphics[width=0.49\textwidth]{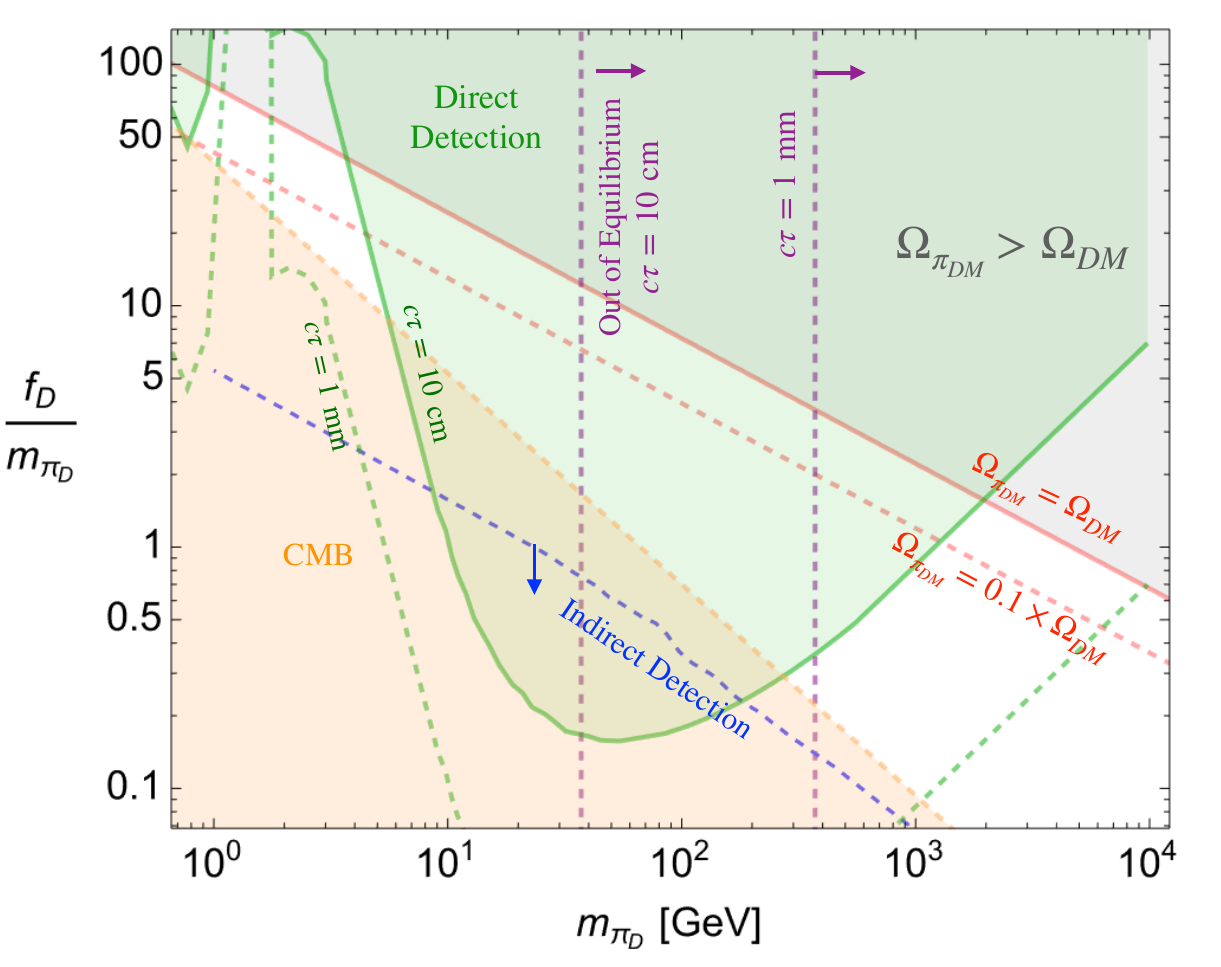}
	\caption{
 Current constraints on sneaky DM. The red line indicates the region of the parameter space in which the entire relic abundance is composed of stable dark pions, while the red dashed line corresponds to the parameter space where they make up 10\% of the total DM. The gray region is excluded because of the over closure of the universe, where only the $2 \to 2$ process is considered. The region below the dashed blue line is ruled out from indirect detection~\cite{Hutten:2022hud}, and the orange region below the dashed orange line is excluded by CMB bounds ~\cite{Planck:2018vyg}. In the parameter space to the right of the purple lines, dark sector is no longer in equilibrium with the SM bath. Therefore, the relic abundance calculation may need to be altered. The area above the green (dashed) line is excluded from direct detection assuming $c \tau =10\,$cm ($c \tau =1\,$mm) for the transient dark pions ~\cite{DirectDetection,SENSEI:2023zdf}. The left (right) plot corresponds to the case where dark sector couples to up (down) type quarks. The sudden fluctuations in direct detections corresponds to the available phase space of transient dark pion.     }
\label{fig:relicDD}
		\end{center}

	\end{figure}

Fig.~\ref{fig:relicDD} shows the current astrophysical and cosmological constraints on the sneaky DM model in the $m_{\pi_D}-f_{D}/m_{\pi_D}$ plane. The red line represents the region of the parameter space where the whole relic abundance is explained by stable dark pions. The gray region above is excluded from the over closure of the universe (where only the 2 → 2 process is considered), whereas the dashed red line corresponds to sneaky DM, which constitutes 10\% of the total relic abundance. As previously stated, we have explicitly confirmed that for the parameter space of interest, 3 → 2 processes are subleading to the 2 → 2 ones. The regions above the green (dashed) line is excluded from direct detection~\cite{DirectDetection,SENSEI:2023zdf} if one assumes for the transient dark pions  $c \tau =1\,$cm ($c \tau =0.1\,$mm). Finally, the region below the dashed blue line is excluded by indirect detection searches~\cite{Hutten:2022hud} whereas the orange region below the orange dashed line is excluded by CMB bounds~\cite{Planck:2018vyg}.

We can see from the figure that stable dark pions with masses $m_{\pi_D}\approx \mathcal{O}(1)\,$GeV and $c\tau=O(1)\,$cm can accommodate the entire relic abundance while satisfying all existing astrophysical and cosmological constraints. Heavier masses $m_{\pi_D}\gtrsim \mathcal{O}(1000)\,$GeV may be allowed from various DM detection experiments, however, one needs to do a more careful analysis of the relic abundance. In the low-mass window, the transient dark pions had to be moderately long-lived to account for the observed relic abundance. We focused on this particular mass window when studying the collider phenomenology of the model because it can provide a light enough mediator to leave observable imprints on the Large Hadron Collider (LHC).

\subsection{Flavour phenomenology}

The flavour phenomenology, including flavour violating meson decays and neutral meson mixing, of this type of models has been discussed in detail for example in~\cite{Renner:2018fhh,Carmona:2021seb,Jubb:2017rhm,Agrawal:2014aoa}. As we consider dark pion masses $\gtrsim \mathcal{O}(1)$~GeV, here we only show bounds from flavour violating $B \to K \pi_D$ decays. We distinguish between invisible final states, where all dark pions are stable on detector scales, and hadronic final states, where dark states decay back to SM. Following~\cite{Carmona:2021seb} the bound in the invisible case is obtained from CLEO data~\cite{CLEO:2001acz} requiring a dark pion lifetime $c\tau > 4$~m for both diagonal and off-diagonal dark pions~\cite{Dolan:2014ska}. For the hadronic final state  the limit $\mathrm{Br}(B\to s g) < 6.8\%$ from~\cite{CLEO:1997xir} is recast. 
 To ensure the dark pion decays inside the detector we require $c\tau < 45$~mm. These constraints are mostly relevant for couplings of dark pions to down-type quarks. An equivalent search for $D\to \pi\, invisible$ does not exist. Finally, another strong constraint arises from neutral meson mixing. For both coupling choices this sets bounds $m_X \gtrsim \mathcal{O}(10)~\text{TeV}$ assuming $\kappa_{\alpha,i}=1$ ($\alpha,i = 1,2,3$). These bounds can be evaded by choosing the respective off-diagonal couplings small~\cite{Renner:2018fhh}. Using non-degenerate couplings also leads to non-degenerate lifetimes of the transient off-diagonal dark pions. As this should not impact the collider phenomenology strongly, for illustration purposes we only show the case of degenerate couplings. Nonetheless, we do not include the neutral meson mixing bounds in Figs.~\ref{fig:mx_mpi1},~\ref{fig:mx_mpi1_down},~\ref{fig:mx_mpi15} and \ref{fig:mx_mpi15_down} as they can be evaded, but show the $B$ decay limits to give an estimate where the flavour bounds would be.

\section{Collider phenomenology}
\label{sec:collider}

The rich particle spectrum of the model can also be searched for at hadron colliders. Pairs of dark quarks can be produced either directly via a t-channel mediator exchange or together with SM quarks in mediator decays. Independent of the production process, dark quarks, analogously to SM quarks, undergo a parton shower and hadronization process in the dark sector, leading to the well known dark shower~\cite{Albouy:2022cin,Strassler:2006im,Han:2007ae,Butterworth:2023cgz} signatures. Different from the original emerging jets scenario~\cite{Schwaller:2015gea}, dark matter is now efficiently produced in the dark shower, such that a significant amount of missing energy is now present in the dark showers. Below we discuss the impact of this feature on different search strategies and propose a set of benchmark scenarios for which we obtain the current experimental constraints. In what follows we focus on low DM masses, $\mpid\leq 100$~GeV.

\subsection{Dark showers from sneaky dark matter}
\begin{figure}
    \centering
    \includegraphics[width=0.85\linewidth]{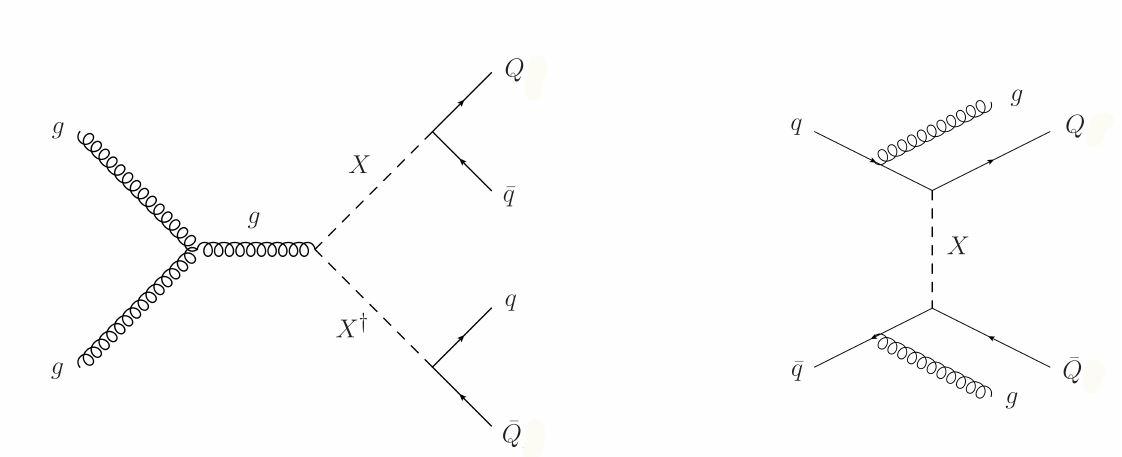}
    \caption{Exemplary Feynman diagrams for the production of two jets and two dark jets from mediator pair production (left) and t-channel mediator exchange (right).}
    \label{fig:feyn_production}
\end{figure}
The two main production processes of dark quarks are pair production of mediators and the production of two dark quarks together with up to two SM jets. Examples of Feynman diagrams that contribute to these processes are displayed in Fig.~\ref{fig:feyn_production}.

The respective cross sections are shown as a function of the mediator mass in Fig.~\ref{fig:production} for $\kappa = 1$ (solid lines) and $\kappa=0.1$ (dashed lines), for a dark quark mass of $m_{Q}=1$~GeV. For the production of dark quarks via t-channel mediator exchange processes with none, one or two additional SM jets were taken into account. All cross sections were calculated using \MADGRAPH~\cite{Alwall:2014hca,Frederix:2018nkq} using a custom UFO model file available at Ref.~\cite{model}.
Note that the production of two dark quarks with up to two SM jets includes both processes with pair produced mediators and processes where the dark quarks are produced via t-channel mediator exchange. 
Furthermore a $p_T$ cut of $p_T \geq 20$~GeV was imposed on the dark quarks to make the cross section well defined also in the $m_{Q}\to 0$ limit.

As expected, the cross section for mediator pair production (left) decreases strongly with $\mx$ from $\sim1$~pb at $\mx = 1$~TeV to $\sim10^{-11}$~pb for $\mx=5$~TeV, whereas the dark quark production cross section only decreases by roughly a factor $\sim500$ over the same mass range. On the other hand, the impact of decreasing $\kappa$ is less significant for mediator pair production, since there is always a contribution which purely depends on the QCD coupling. This leads to another noteworthy feature: For small enough $\kappa$ the mediator pair production cross section becomes independent of the choice of mediator hypercharge. For $\kappa\sim1$ couplings to up-type quarks (blue) lead to slightly larger cross sections than couplings to down-type quarks (orange) due to protons containing two up quarks leading to a larger contribution from t-channel dark quark exchange. As this contribution becomes more and more suppressed for smaller $\kappa$ the cross sections for couplings to up- and down-type quarks become the same. For the t-channel mediator exchange production of two dark quarks, however, the cross section is proportional to $\kappa^4$. Thus, once the mediator mass is large enough so that resonant mediator production becomes irrelevant, the two dark quark plus up to two SM jets production cross section decreases as $\kappa^4$. 

\begin{figure}
    \centering
    \includegraphics[width=0.47\linewidth]{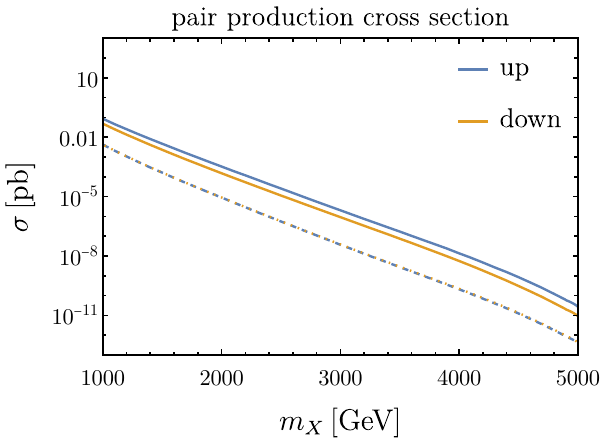}
    \quad
    \includegraphics[width=0.47\linewidth]{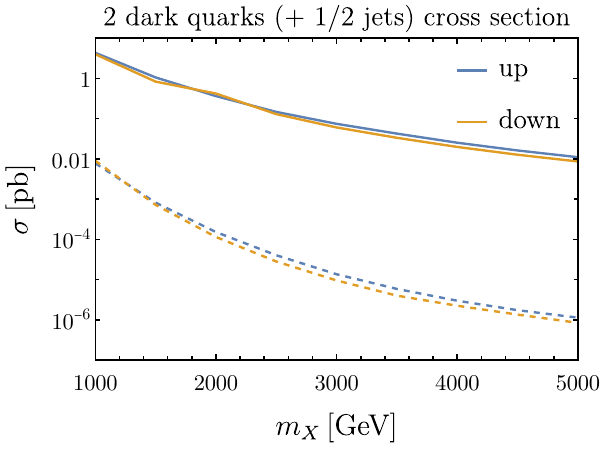}
    \caption{Cross sections for pair produced mediators (left) and the production of two dark quarks via t-channel mediator exchange with up to two additional SM jets (right) for couplings to up- (blue) or down-type quarks (orange) with $\kappa=1$ (solid) and $\kappa=0.1$ (dashed). A $p_T \geq 20$~GeV cut is included for dark quark production. 
    }
    \label{fig:production}
\end{figure}

A striking consequence of the production of dark quarks is that they undergo showering and hadronization in the dark sector, leading to spectacular signatures such as emerging jets~\cite{Schwaller:2015gea} or semi-visible jets~\cite{Cohen:2020afv}, which are now searched for by the ATLAS and CMS experiments~\cite{CMS:2018bvr,CMS:2021dzg,ATLAS:2022aaw,CMS:2024gxp}. Since the phenomenology of dark showers is discussed at length in the literature (see e.g.~\cite{Albouy:2022cin,Butterworth:2023cgz} for reviews, and \cite{Carrasco:2023loy} for recent related work), we focus here on the aspects which are characteristic of our scenario. 
First it should be noted that in the hard processes depicted in Fig.~\ref{fig:feyn_production}, only the first three generations of dark quarks can be produced. The fourth flavour
$Q_{4}$ can only originate from gluons splitting into same generation quark pairs in the parton shower. It follows that only an even number of
stable dark pions, which are combinations of $Q_4$ and $Q_i$ with $i \neq 4$, can occur in each event. We have verified that the implementation of the parton shower and hadronization in \PYTHIA takes this into account properly. 

The fraction of stable dark pions determines how much missing energy is on average contained in the dark jets. The ratio of diagonal to off-diagonal degrees of freedom for four dark flavours is $(4-1)/(4^2-1-(4-1))=1/4$. However the Lund string fragmentation model implemented for the HiddenValley module of \PYTHIA  gives a ratio of $1/(4-1)=1/3$. 
Combined with the fact that there are three stable and three transient (complex) off-diagonal dark pions the average fraction of stable dark pions in each event should be $2/3\times 1/2 = 1/3$ as long as the transient dark pions decay promptly. If the transient dark pions have lifetimes of order the detector scale this value increases. It follows that there will now always be a significant amount of missing energy (MET) in the signal events, even in the limit where the transient dark pions decay promptly. 

Let us now consider the production of two dark and two SM jets through pair produced mediators. Then, depending on the transient dark pion lifetime several possible signatures can be distinguished: 
If the transient dark pions decay promptly the signal consists of four jets, while for long-lived transient dark pions two jets and MET can be seen. In the intermediate lifetime regime, when dark pions decay back to SM particles inside the detector the spectacular signature of emerging jets~\cite{Schwaller:2015gea} arises. Due to different dark pion momenta the dark pions in a jet decay to SM particles at various points so that a jet "emerges". All three signatures have been searched for at ATLAS and/or CMS~\cite{ATLAS:2020syg,CMS:2018bvr,CMS:2019zmd,CMS:2022usq}.  Note that these signatures also arise when a pair of dark quarks is produced via t-channel mediator exchange together with 2 initial state radiation (ISR) jets. We discuss the implications of this production channel in more detail below. In addition, ATLAS performed a dedicated search for pairs of semi-visible jets, jets that consist of visible and invisible particles, produced from t-channel mediator exchange with up to two SM jets~\cite{ATLAS:2022aaw}. In the following, all four signatures are used to constrain the parameter space and to see if the dark matter region in Fig.~\ref{fig:relicDD} can be probed at colliders.

\subsection{Phenomenological Parameters}
As usual for composite models, a sizeable number of parameters are needed to fully specify the model. Fortunately in practice, the following combination of fundamental and effective parameters are sufficient to determine the collider phenomenology:   
\begin{align}
    \mx \;, \quad \kappa \;,  \quad m_{\pi_D} \;, \quad f_{D} \;,\quad  \Lambda_D \;,\quad m_{Q} .
\end{align}
Here $\mx$ and $\kappa$ are the mediator mass and Yukawa coupling matrix appearing in the Lagrangian, $\Lambda_D$ is the dark confinement scale, and $m_{\pi_D}$ and $f_{D}$ are the mass and decay constant of the dark pions that govern their effective interactions. Finally $m_{Q}$ is not the actual dark quark mass but should rather be understood as a constituent mass used as an input for \PYTHIA's HiddenValley module. 

Instead of $\kappa$, it is convenient to use the dark pion lifetimes $c\tau_{\pi_{ij}}$ as free input parameters. If we further set $\kappa_{\alpha i}=\kappa$ for $\alpha, i = 1,2,3$ ($\kappa_{4 i}=0$), all diagonal dark pions have the same lifetime, and the same is true for transient off-diagonal dark pions. Furthermore the diagonal and off-diagonal dark pion lifetimes are connected via
\begin{align}
    c\tpid^{\rm off-diag}=c\tpid^{\rm diag}\times \frac{3}{4},
\end{align}
leaving us with only one lifetime parameter. 

Moreover, in \PYTHIA we choose 
$\frac{1}{2}m_{\rho_D}=2m_{\pi_D} = m_{Q}$ for the phenomenological parameters.  
The remaining free quantities are $m_X$, $\mpid$, $f_D$ and $\Lambda_D$. Of these, the last three are in general not independent, but determined by the strong dynamics. From naive dimensional analysis we expect  $f_D \approx \Lambda_D/4\pi$. While for $f_{D}\approx m_{\pi_D}$ this is never grossly violated with $\Lambda_D = 2m_{\pi_D}$ in the mass range of interest, from a dark matter perspective we are more interested in the case $f_{D}=\mathcal{O}(10)m_{\pi_D}$. To ensure that the relation above is not grossly violated, $\Lambda_D=40m_{\pi_D}$ is used for $f_D \sim\mathcal{O}(10)\mpid$. Since the efficiency of collider searches does not depend strongly on the mediator mass~\cite{Mies:2020mzw}, we use $m_X=2$~TeV for event generation in all recasts, which we perform for dark pion masses $1\leq \, \mpid \leq 50~\text{GeV}$ and lifetimes $10^{-4} \leq\, c\tpid^{\rm diag} \leq 10^4~\text{mm}$. 

In the semi-visible jets search~\cite{ATLAS:2022aaw}, the fraction of invisible dark pions, $R_{\rm inv}$ appears as another model parameter. In our scenario there is a lower bound on this quantity due to DM pions being produced in the shower, however it also depends on the lifetime of the transient dark pions and can thus take values ranging from $\sim0.3$ for prompt decays to $\sim0.95$.

\subsection{Collider limits and future projections}
Searches with the full run 2 data set of $139$~fb$^{-1}$ have been performed for non-resonant  dijet pairs~\cite{CMS:2022usq} at CMS and for two jets plus missing energy at both CMS~\cite{CMS:2019zmd} and ATLAS~\cite{ATLAS:2020syg}. Dedicated searches for semi-visible jets~\cite{ATLAS:2022aaw} at ATLAS and for emerging jets~\cite{CMS:2018bvr,CMS:2024gxp} at CMS using the full run 2 data set have also appeared recently. In the following, we describe how these searches are used to place constraints on the model parameter space.

The CMS search for jets plus missing energy~\cite{CMS:2019zmd} was performed as a search for supersymmetry with the lightest neutralino as a dark matter candidate and considering several neutralino masses. To set bounds on $\mx$ we use the cross section limit for the lightest given neutralino mass of $\sim 6$~GeV. The lowest non-zero neutralino mass used in the corresponding ATLAS search is with $100$~GeV well above the mass region we are interested in, so that we only use the CMS search in the following. 
For completeness, we show the cross section limits from the CMS two jet plus MET search~\cite{CMS:2019zmd}, from the CMS non-resonant pairs of dijets search~\cite{CMS:2022usq} and from the ATLAS semi-visible jets search~\cite{ATLAS:2022aaw} in Fig.~\ref{fig:xseclimit} in appendix~\ref{app:xsec}, where we also discuss more details of the event selection. For the CMS emerging jet search cross section limits are given directly in the $c\tpid-\mx$ - plane for mediator masses $400 \le \mx \le 2000$~GeV and lifetimes $1~{\rm mm} \le c\tpid \le 1000$~mm and are not shown here.

To find the upper limit on the mediator mass from the four jets search, the two jet plus missing energy and semi-visible jets search we use the \PYTHIA~\cite{Bierlich:2022pfr} HiddenValley module on the events generated by \MADGRAPH  for $m_X = 2$~TeV to carry out the dark sector and SM showering and hadronization. Hereby, dark pion mass (and the related parameters) and lifetimes were varied in the ranges
\begin{align}
    1\leq \, &\mpid \leq 50~\text{GeV}\,,\\
    10^{-4} \leq\, &c\tpid^{\rm diag} \leq 10^4~\text{mm}\,.
\end{align}
For each dark pion mass and lifetime combination the efficiency of all three searches is found by applying the same selection criteria as in the experimental searches. The selection criteria are listed in Tab.~\ref{table:selection} and are discussed in more detail in appendix~\ref{app:selection}.  

The resulting efficiencies as a function of the lifetime  are shown in Fig.~\ref{fig:efficiency} for $\mpid=1$~GeV for four jets, two jets plus MET and semi-visible search for couplings to up-type quarks, while it is shown for $\mpid = 1$ and 10~GeV for the semi-visible search for couplings to down-type quarks. The efficiencies for couplings to up-type quarks are shown as dashed lines, for couplings to down-type quarks as solid lines. There are several interesting, but not unexpected features: While the jets plus missing energy is most efficient for long-lived dark pions, it remains efficient even for smaller lifetimes due to the production of dark matter pions in the dark shower. The four jet search efficiency instead decreases for larger lifetimes as fewer dark pions contribute to $H_T$. In addition, it is suppressed by the requirement to reconstruct two heavy states. Finally the semi-visible jets search increases with lifetime as the missing energy requirement gets fulfilled more frequently, but is suppressed due to the low momenta of jets from t-channel mediator exchange. It also applies a b-jet veto, which suppresses it further for $\mpid > m_b$ for couplings to down type quarks. 
\begin{figure}
    \centering
    \includegraphics[width=0.6\linewidth]{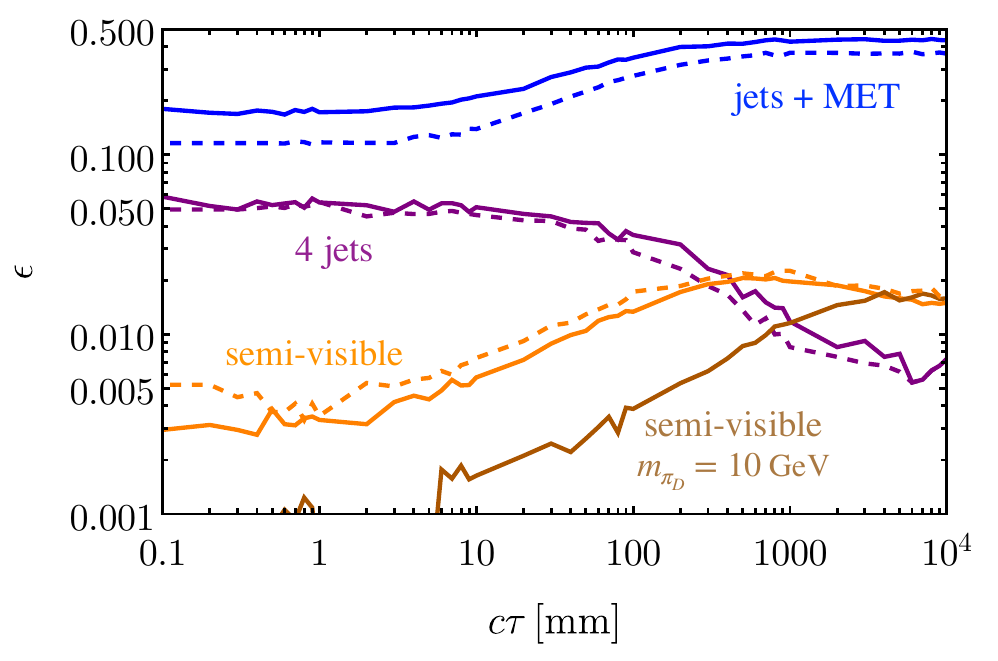}
    \caption{Efficiencies of the two jets plus MET (blue), four jet (purple) and semi-visible jets (orange) searches for $m_{\pi_D} = 1$~GeV, $m_X = 2$~TeV, $\mpid=f_D$ and $\kappa = 1$ as a function of the diagonal dark pion lifetimes. The lifetime of the off-diagonal dark pions has been adjusted accordingly. Couplings to up-type quarks are shown by the dashed lines, couplings to down-type quarks by solid lines. For couplings to down-type quarks for the semi-visible search the efficiencies for $m_{\pi_D}=10\,\mathrm{GeV}$ are shown in brown as well. 
    }
    \label{fig:efficiency}
\end{figure}

Before using the acceptances to set limits on the mediator mass it is important to point out that dark hadronization, like SM hadronization, is a non-perturbative process that cannot be calculated, but has to modeled. In \PYTHIA, Lund string fragmentation is used to model SM and dark sector hadronization. Unlike for SM hadronization no data exist for dark sector hadronization that could be used to set the Lund string fragmentation model parameters. For this study we use the default values implemented in \PYTHIA for the dark sector hadronization. However, as there is no theory prior for these values, these values could be very different. While this can have significant effects on observables like shower shapes and jet substructures~\cite{Cohen:2020afv,Cohen:2023mya}, searches which do not make use of these observables, such as the ones considered here, should not be strongly affected.

Based on the production cross sections and search efficiencies we can calculate fiducial cross sections
\begin{align}
    \sigma(\mpid,c\tpid) = \epsilon_{\rm eff}(\mpid,c\tpid)\times \sigma_{\rm prod}\,
\end{align}
for the two jets plus MET, four jets and semi-visible searches and compare them to the observed cross section limits to find upper limits on the mediator mass. For the semi-visible search the cross section limit corresponding to the closest $R_{inv}$ value is used. 
For the emerging jet search, on the other hand, we use the efficiencies and number of background events provided in~\cite{CMS:2024gxp} and use the procedure described in~\cite{Mies:2020mzw} to set upper limits on the mediator mass. Unlike in~\cite{Mies:2020mzw}, where the diagonal and off-diagonal dark pion lifetimes are varied independently, here, the diagonal dark pions lifetimes are set to the values suggested in~\cite{CMS:2024gxp}, while 
the off-diagonal dark pion lifetimes are calculated following \eqref{eq:partdec}. The efficiencies and number of background events are taken according to the diagonal dark pion lifetime. This is a less conservative choice than used in~\cite{Mies:2020mzw}, meaning the emerging jet limits might be slightly over-estimated in some regions. 

The limits on the mediator mass are shown as a function of the lifetime in Fig.~\ref{fig:mX_limit} for $\mpid=1$~GeV for all four searches (and for $\mpid = 10$~GeV for the semi-visible search with couplings to down-type quarks), $\mpid= f_D$ and $\kappa=1$. It can be seen that for small lifetimes the four jets search (purple) puts the strongest constraint $m_X \lesssim1500$~GeV for couplings to up-type quarks, while for couplings to down-type quarks it is the jets plus MET search (blue) and the four jets search, both set a limit close to $m_X = 1200$~GeV. While the efficiencies for the jets plus MET search is higher, the bounds on the four jets cross sections are more stringent leading to a more stringent constraint for couplings to up-type quarks, where the production cross section is larger, and comparable limits for couplings to down-type quarks. In the intermediate region the emerging jets search sets the strongest limit, up to $\mX \sim 2250$~GeV (2000~GeV) for up-type\footnote{In the up-type case, decays $X\to t Q_D$ are included, which could suffer from a lower acceptance, presumably leading to a slight over-estimation of the limit.} (down-type) couplings. A possible further improvement could be the inclusion of t-channel mediator exchange production of dark quarks, in which case masses up to $\mX=2500$~GeV could be probed.

\begin{figure}
    \centering
    \includegraphics[width=0.48\linewidth]{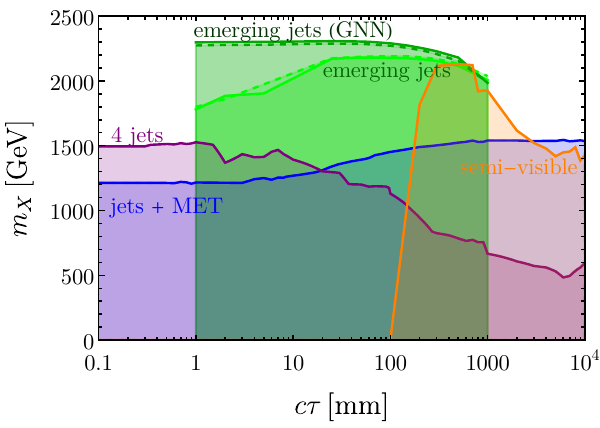}
     \quad
     \includegraphics[width=0.48\linewidth]{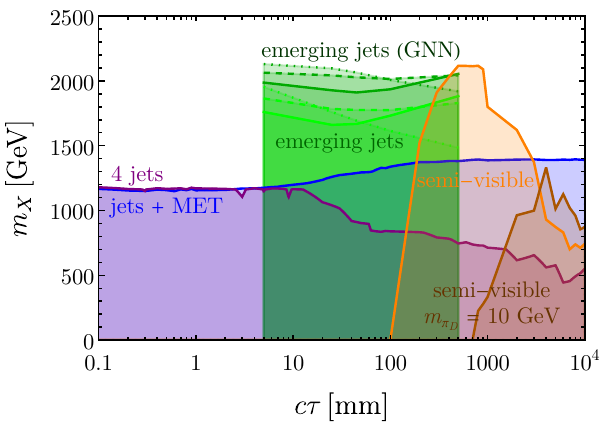}
    \caption{Upper limit on $m_X$ from the emerging jet search (solid, dashed and dotted green and dark green show the limits from the model agnostic and GNN search with the full run 2 data set for $\mpid = 10,\, 20,\, 6$~GeV), MET search (blue), four jets search (purple) and semi-visible search (orange, brown for $\mpid=10$~GeV) for couplings $\kappa=1$ and $\mpid=1$~GeV to up-type (left) and down-type quarks (right). }
    \label{fig:mX_limit}
\end{figure}
The semi-visible search sets the strongest limits up to $m_X\sim 2100$~GeV for lifetimes around 1~m where the emerging jet searches sensitivity decreases. Finally, for large lifetimes the MET plus jets search sets the strongest limit at $\mX \sim 1500$ (1400)~GeV for couplings to up- (down-)type quarks. This is to be expected, as this search is designed for this exact scenario.

The same is show in Fig.~\ref{fig:mxlimit01} in appendix~\ref{app:xsec} for $\kappa=0.1$. As anticipated above, the strong $\kappa$ suppression of the t-channel mediator exchange means the fiducial cross section for the semi-visible search is always below the experimental limit, while for the other three searches the limits weaken to $\sim 500-600$~GeV in the small lifetime region, $\sim 1500$~GeV in the intermediate lifetime region and $\sim 600-700$~GeV in the large lifetime region. These results actually highlight the power of dedicated long lived particle searches like the emerging jets search. Since they suffer from lower backgrounds, they remain sensitive to significantly reduced cross sections, and this could further improve in future high luminosity runs.

\subsection{Combined limits}

Finally it is important to understand to which extent the collider limits are able to probe the parameter space where the dark matter pions are viable DM candidates. This is difficult in practice, since the various probes are usually sensitive to different combinations of parameters, and for the sake of presentation, some of them must be set to fixed values. 

For DM searches at colliders, constraints are most commonly prepared in the parameter plane spanned by the masses of the DM and the mediator, and we follow this approach here. This leaves $\kappa$ and $f_D$ undetermined. In Figs.~\ref{fig:mx_mpi1}, \ref{fig:mx_mpi1_down} we show results for $\kappa=1$ and $\kappa = 0.1$ with $f_D = m_{\pi_D}$. Fixing these parameters essentially relates the lifetime to the DM mass, such that the collider limits appear rotated by 90 degrees (and flipped upside down). Direct detection constraints (shown in grey) overtake the collider constraints for DM masses above $5-10$~GeV.  For smaller $\kappa$ both collider and direct detection limits weaken, and the region probed by the emerging jets search moves to larger DM masses, in order to keep the lifetime unchanged. For down-type quark couplings, the parameter space is further constrained by searches for flavour violating $B$-meson decays. In particular the $B\to K \; invisible$ search probes very large mediator masses, into the tens of TeV range. Potentially a dedicated search for $B\to K + long\;lived\;particle$ could fully probe the parameter space in between the collider and flavour bounds.

\begin{figure}
    \centering
    \includegraphics[width=0.47\linewidth]{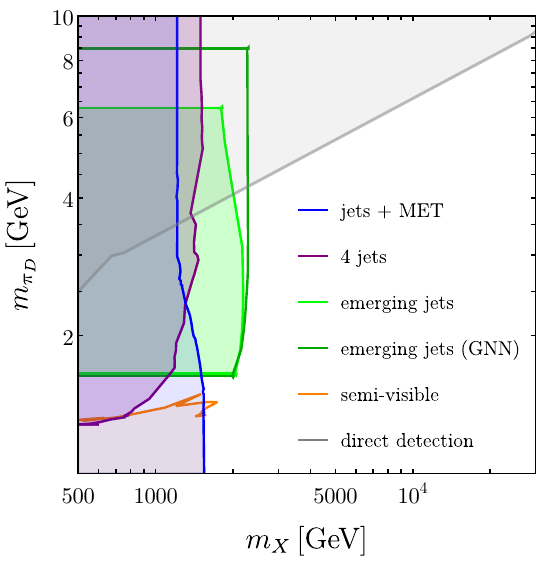}\quad
    \includegraphics[width=0.47\linewidth]{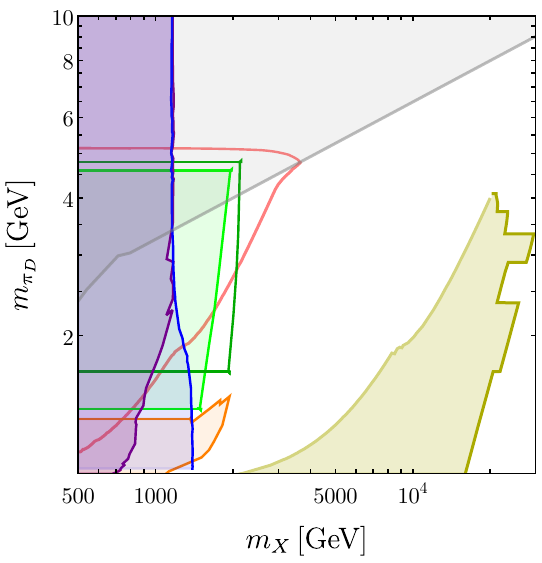}
    \caption{Constraints in the $\mpid-\mX$ frame for $\mpid=f_D$ and coupling $\kappa = 1$  left (right) to up-type (down-type) quarks. Purple, blue, (dark) green and orange regions are excluded from four jets, jets plus missing energy, (GNN) model agnostic emerging jets and semi-visible jets search. The grey region is excluded from direct detection experiments. In addition estimated bounds from prompt flavour-violating $B$ decays are shown in red and the bounds from $B\to K\,invisible$ are shown in olive for couplings to down-type quarks. }
    \label{fig:mx_mpi1}
\end{figure}

\begin{figure}
    \centering
    \includegraphics[width=0.47\linewidth]{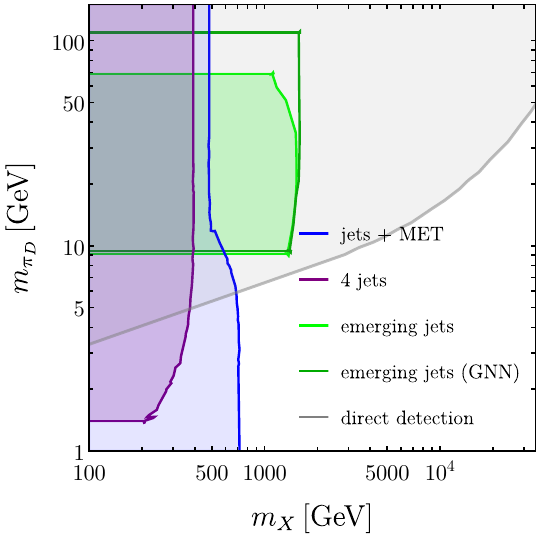}\quad
    \includegraphics[width=0.47\linewidth]{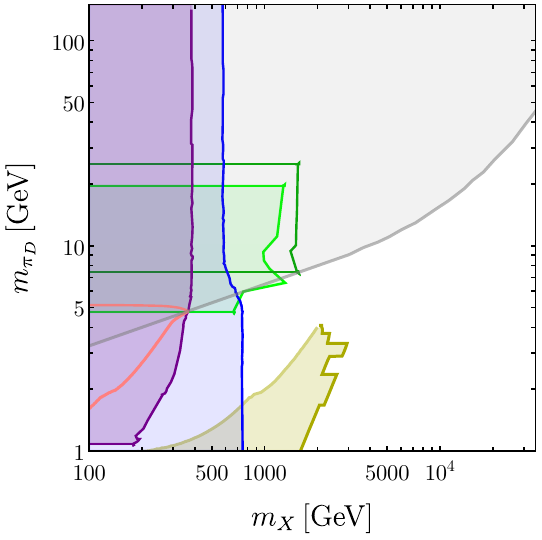}
    \caption{Same as figure \ref{fig:mx_mpi1}, but for $\kappa=0.1$.}
    \label{fig:mx_mpi1_down}
\end{figure}

In addition in appendix~\ref{app:xsec} results for $f_D = 15 m_{\pi_D}$ are shown. It is clear that in this case a smaller dark pion mass region is excluded by collider searches than for $f_D=\mpid$, due to the additional factor $15^2$ that now appears in the expressions for the lifetimes. Since this parameter choice is preferred from a relic abundance perspective, covering the region down to $m_{\pi_D}\sim 1$~GeV is important, while direct detection takes over already between $2-4$~GeV. 

Overall, we see that direct detection, colliders and flavour experiments probe complementary parameter regions of sneaky DM. Let us finish the section by speculating on the prospects of future searches. Clearly a future hadron collider operating at higher collision energies can push the collider limits to higher mediators masses, and possibly probe the $5-10$~TeV region. A multi-TeV lepton collider can pair produce the mediators since they are electrically charged, and should be able to push the limits close to half the collider energy. In addition, radiation of a single mediator from a quark, converting it into a dark quark, or the loop induced pair production of dark quarks could lead to interesting signatures that could be explored in the future. Dark matter direct detection experiments have reached the neutrino fog~\cite{XENON:2024ijk} at recoil energies corresponding to $\sim 10$~GeV DM masses, and are getting close for higher masses. In the future one can therefore expect more incremental improvements in the sensitivity. For the best motivated region of DM masses in the $1-10$~GeV region, there is however still a lot of room for improvement. Searches for rare meson decays could probe DM masses below the B meson mass, but only in the case of down-type couplings, again highlighting the importance of a comprehensive search program.

\section{Conclusions}
The origin of DM constitutes one of the most important problems in particle physics. However, in practice, we know very little about its nature, and very distinct paradigms are still compatible with existing experimental searches.  The long-time leading DM candidate, a weakly interacting massive particle, is not at its best due to the absence of new physics at colliders and in DM direct detection experiments. A possible way out of this impasse is that (light) new physics associated with DM is long-lived, since most of the searches so far assumed prompt decays. This is even more important for light new physics, since the lower the energy the more the data there is available. In this work, we have presented a model, or class of models, featuring viable DM candidates in the GeV range coming along with long-lived particles and a very rich collider phenomenology. 

More specifically, we have presented a fairly minimal composite model of DM based on a QCD-like dark sector featuring a $SU(N_d)$ gauge symmetry, $n_f$ dark quarks, and a heavy t-channel scalar mediator.  We have demonstrated that for $n_f\ge 4$, the dark flavor symmetry $\mathcal{G}\equiv SU(n_f-3)\times U(1)$ guarantees the stability of a subset of the dark pions, which are our DM candidates. Their interaction with the rest of the dark pions (that we dubbed transient dark pions) and their subsequent decay to SM quarks controls the DM relic abundance as well as the cross-section in direct and indirect detection experiments, which become  independent of the mass and couplings of the heavy mediator. Moreover, along the lines of models of impeded DM, due to the degeneracy of the dark pions and their pNGB nature, the co-annihilation cross-section depends linearly on the velocity and is suppressed at low temperatures, alleviating both direct and indirect experimental bounds and opening up the GeV-mass window for this \emph{sneaky} DM, typically excluded for 
thermal production mechanisms.

For the specific case of $n_f=4$ (where the dark sector is comprised of 9 transient and 6 stable dark pions), we have demonstrated that the sneaky DM can be as light as a few $\mathrm{GeV}$, without contradicting current bounds from direct and indirect DM detection. For such a window, the transient dark pions are long-lived with $c\tau\gtrsim 1\,\mathrm{mm}$, leading to spectacular signals at colliders such as emerging or semi-visible jets. In particular, we have derived constraints for the mediator mass in the case of dark matter masses $\mathcal{O}(1)\,\mathrm{GeV}$ using several existing searches by ATLAS and CMS and estimated the reach of future experiments. For fixed values of the dark meson masses and the transient dark pion decay constant $f_D$, we have set bounds on the mediator masses as a function of the transient dark pions lifetime combining searches for four jets, two jets + MET, semi-visible jets as well as searches for emerging jets or bounds from flavour violating meson decays, excluding $\mathcal{O}(1)\,\mathrm{TeV}$ mediators for order one  values of the portal coupling. 

Compared with similar models that give rise to dark showers at colliders, sneaky dark matter predicts a dark shower composed of both stable and long lived dark pions, which suggests a search that combines both emerging jets and semi-visible jets strategies. The model is minimal in the sense that no additional particles or mechanisms are required to explain dark matter, and thus represents a key benchmark model for collider searches for composite dark matter. Finally the DM is produced copiously in the shower - in the case of a detection, this will therefore open the possibility for probing dark matter properties with precision studies of dark showers. Possible future directions include studying the prospect of detecting sneaky DM at future colliders and in fixed target/beam dump experiments, as well as more careful studies of the cosmic ray spectra produced by our model.

 \appendix
\section*{Acknowledgments} We thank J. Smirnov for pointing out the importance of the annihilation of stable dark baryons,  S. Sinha for encouraging discussions regarding searches for the model at the colliders, K. Pedro for useful comments on the draft, and all participants of the MITP program ``The Dark Matter Landscape" for providing a stimulating atmosphere.  AC thanks Renato Fonseca for illuminating  discussions on several group theory aspects of this work. The work of AC has been partially supported by MCIN/AEI (10.13039/501100011033) and ERDF (grant PID2022-139466NB-C21), as well as by the Junta de Andalucía grants FQM 101 and P21\_00199. CS is supported by the Office of High Energy Physics of the U.S. Department of Energy under contract DE-AC02-05CH11231 and through the Alexander von Humboldt Foundation. FE and PS are supported by the Cluster of Excellence Precision Physics, Fundamental Interactions and Structure
of Matter (PRISMA+ EXC 2118/1) funded by the German Research Foundation (DFG) within
the German Excellence Strategy (Project No. 390831469).
We would like to acknowledge the Mainz Institute for Theoretical Physics (MITP) and the CERN Theoretical Physics Department for hospitality and partial support during the completion of this work.

 \section{$SU(n_f)$ group theory}
 \label{app:sun}
 We take $SU(n_f)$ generators with the following normalization 
\begin{align}
    \mathrm{Tr}\big(T^a T^b\big)=\frac{1}{2}\delta^{ab}.
\end{align}
We can define the tensors $f^{abc}$ and $d^{abc}$
\begin{align}
    \big[T^a,T^b\big]=if^{abc}T^c,\quad d_{abc}=2\mathrm{Tr}\Big(\big\{T^a,T^b\big\}T^c\Big)
\end{align}
which are totally anti-symmetric and totally symmetric, respectively. Then~\cite{Haber:2019sgz}
\begin{align}
    T^aT^b=\frac{1}{2}\Big[\frac{1}{n_f}\delta^{ab}\mathbf{1}+(d^{abc}+if^{abc})T^c\Big]
    \label{eq:twots}
\end{align}
which leads to
\begin{align}
    \mathrm{Tr}\Big[T^aT^b T^c T^d\Big]=\frac{1}{4n_f}\delta^{ab}\delta^{cd}+\frac{1}{8}\big(d^{abe}d^{cde}-f^{abe}f^{cde}+if^{abe}d^{cde}+if^{cde}d^{abe}\big).
\end{align}
This relationship can also be written as
\begin{align}
    \mathrm{Tr}\Big[T^aT^b T^c T^d\Big]&=\frac{1}{4n_f}\big(\delta^{ab}\delta^{cd}-\delta^{ac}\delta^{bd}+\delta^{ad}\delta^{bc}\big)+\frac{1}{8}\big(d^{abe}d^{cde}-d^{ace}d^{bde}+d^{ade}d^{bce}\big)\nonumber\\
    &+\frac{i}{8}\big(d^{abe}f^{cde}+d^{ace}f^{bde}+d^{ade}f^{bce}\big).
\end{align}
Finally,
\begin{align}
\mathrm{Tr}\Big[T^aT^bT^cT^dT^e\Big]&=\frac{1}{8n_f}\big(d^{abe}+if^{abe}\big)\delta^{cd}+\frac{1}{8n_f}\delta^{ab}\big(d^{cde}+if^{cde}\big)\nonumber\\
&+\frac{1}{16}(d^{ab\alpha}+if^{ab\alpha})(d^{cd\beta}+if^{cd\beta})(d^{\alpha \beta e}+if^{\alpha \beta e}).
\end{align}

The pNGBs associated to the spontaneous symmetry breaking of $SU(n_f)_L\times SU(n_f)_R\to SU(n_f)_V$ will transform in the adjoint of $SU(n_f)_V$, which decomposes under its maximal subgroup $SU(3)\times SU(n_f-3)\times U(1)\subset SU(n_f)$ as
\begin{align}
\mathbf{Adj}=(\mathbf{8},1)_0\oplus(\mathbf{1},\mathbf{(n_f-3)^2-1)}_0\oplus (\mathbf{1},\mathbf{1})_0\oplus\left[(\mathbf{3},\overline{\mathbf{n_f-3}})_{\frac{n_f a}{n_f-3}}+\mathrm{h.c.}\right], \quad |a|=\sqrt{\frac{n_f-3}{6n_f}}.
\label{eq:adj}
\end{align}

Through this work, we will use the following generators of $\mathfrak{su}(4)$
\begin{align}
T_1&=\left(
\begin{array}{cccc}
 0 & \frac{1}{2} & 0 & 0 \\
 \frac{1}{2} & 0 & 0 & 0 \\
 0 & 0 & 0 & 0 \\
 0 & 0 & 0 & 0 \\
\end{array}
\right),\quad T_2=\left(
\begin{array}{cccc}
 0 & -\frac{i}{2} & 0 & 0 \\
 \frac{i}{2} & 0 & 0 & 0 \\
 0 & 0 & 0 & 0 \\
 0 & 0 & 0 & 0 \\
\end{array}
\right),\quad T_3=\left(
\begin{array}{cccc}
 \frac{1}{2} & 0 & 0 & 0 \\
 0 & -\frac{1}{2} & 0 & 0 \\
 0 & 0 & 0 & 0 \\
 0 & 0 & 0 & 0 \\
\end{array}
\right),\quad T_4=\left(
\begin{array}{cccc}
 0 & 0 & \frac{1}{2} & 0 \\
 0 & 0 & 0 & 0 \\
 \frac{1}{2} & 0 & 0 & 0 \\
 0 & 0 & 0 & 0 \\
\end{array}
\right),\nonumber\\
T_5=&\left(
\begin{array}{cccc}
 0 & 0 & -\frac{i}{2} & 0 \\
 0 & 0 & 0 & 0 \\
 \frac{i}{2} & 0 & 0 & 0 \\
 0 & 0 & 0 & 0 \\
\end{array}
\right),\quad T_6=\left(
\begin{array}{cccc}
 0 & 0 & 0 & 0 \\
 0 & 0 & \frac{1}{2} & 0 \\
 0 & \frac{1}{2} & 0 & 0 \\
 0 & 0 & 0 & 0 \\
\end{array}
\right),\quad
T_7=\left(
\begin{array}{cccc}
 0 & 0 & 0 & 0 \\
 0 & 0 & -\frac{i}{2} & 0 \\
 0 & \frac{i}{2} & 0 & 0 \\
 0 & 0 & 0 & 0 \\
\end{array}
\right),\, T_8= \frac{1}{2 \sqrt{3}} \left(
\begin{array}{cccc}
1 & 0 & 0 & 0
   \\
 0 & 1 & 0 & 0
   \\
 0 & 0 & -2 & 0 \\
 0 & 0 & 0 & 0 \\
\end{array}
\right),\nonumber\\
T_9&=\left(
\begin{array}{cccc}
 0 & 0 & 0 & \frac{1}{2} \\
 0 & 0 & 0 & 0 \\
 0 & 0 & 0 & 0 \\
 \frac{1}{2} & 0 & 0 & 0 \\
\end{array}
\right),\quad T_{10}=\left(
\begin{array}{cccc}
 0 & 0 & 0 & -\frac{i}{2} \\
 0 & 0 & 0 & 0 \\
 0 & 0 & 0 & 0 \\
 \frac{i}{2} & 0 & 0 & 0 \\
\end{array}
\right),\quad T_{11}=\left(
\begin{array}{cccc}
 0 & 0 & 0 & 0 \\
 0 & 0 & 0 & \frac{1}{2} \\
 0 & 0 & 0 & 0 \\
 0 & \frac{1}{2} & 0 & 0 \\
\end{array}
\right),\quad T_{12}=\left(
\begin{array}{cccc}
 0 & 0 & 0 & 0 \\
 0 & 0 & 0 & -\frac{i}{2} \\
 0 & 0 & 0 & 0 \\
 0 & \frac{i}{2} & 0 & 0 \\
\end{array}
\right),\nonumber\\
T_{13}&=\left(
\begin{array}{cccc}
 0 & 0 & 0 & 0 \\
 0 & 0 & 0 & 0 \\
 0 & 0 & 0 & \frac{1}{2} \\
 0 & 0 & \frac{1}{2} & 0 \\
\end{array}
\right),\quad  T_{14}=\left(
\begin{array}{cccc}
 0 & 0 & 0 & 0 \\
 0 & 0 & 0 & 0 \\
 0 & 0 & 0 & -\frac{i}{2} \\
 0 & 0 & \frac{i}{2} & 0 \\
\end{array}
\right),\quad T_{15}=\frac{1}{2 \sqrt{6}} \left(
\begin{array}{cccc}
1 & 0 & 0 & 0
   \\
 0 & 1 & 0 & 0
   \\
 0 & 0 & 1 & 0
   \\
 0 & 0 & 0 &
   -3 \\
\end{array}
\right).
\end{align}
One can readily see that the first eight generators fulfill the $\mathfrak{su}(3)$ algebra whereas the last one generates a $\mathfrak{u}(1)$. On the other hand, the adjoint of $SU(4)$ decomposes under its maximal subgroup $SU(3)\times U(1)\subset SU(4)$ as
\begin{align}
\mathbf{15}=\mathbf{8}_0\oplus \mathbf{1}_0\oplus\left[\mathbf{3}_{\sqrt{\frac{2}{3}}}+\mathrm{h.c.}\right].
\end{align}

\section{Accidental $\mathbb{Z}_2$ symmetry of the ChPT Lagrangian for $n_f=4$}
\label{ap:sym}
Here we will prove that in the case of $n_f=4$, i.e., $SU(4)_L\times SU(4)_R\to SU(4)_V$, there is an accidental symmetry making $\pi_{\rm DM}$ to appear always in pairs in any term of the Chiral Lagrangian.
\begin{proposition}
\label{teor:su4}
Let us consider a global $SU(4)_L\times SU(4)_R$  symmetry group  spontaneously broken to its diagonal $SU(4)_V$. Such breaking features $n_f^2-1=15$ Nambu-Goldstone bosons that under its $SU(3)\subset SU(4)_V$ subgroup decomposes as $\mathbf{15}=\mathbf{8}\oplus\mathbf{3}\oplus\bar{\mathbf{3}}\oplus \mathbf{1}$. Then, any interaction arising from the Chiral Lagrangian (including the Wess-Zumino-Witten 5-point interactions) can only involve an even number of Nambu-Goldstone bosons transforming in the $\mathbf{3}$ (or its conjugated $\bar{\mathbf{3}}$).
\end{proposition}

        \begin{proof}
If we parametrize the 15 Nambu-Goldstone bosons as 
\begin{align}
U=\mathrm{exp}\bigg[\frac{2i}{f} \Pi_D \bigg],\qquad \Pi_D=\pi^a_D T^a
\end{align}
with $T^a$ the $SU(4)$ generators, with $\mathrm{Tr}(T^a\cdot T^b)=1/2$, there is a basis such that
\begin{align}
\Pi_D=\left[
\begin{array}{c|c}
\mathbf{8} & \mathbf{3} \\ \hline
\mathbf{3}^{\dagger} &\mathbf{1} 
\end{array}\right].
\end{align}
Note that the $\bar{\mathbf{3}}$ is just the hermitian conjugated of the complex $\mathbf{3}$, since we only have 15 real scalars.         On the other hand, any interaction arising from the chiral Lagrangian will be proportional to the trace of a certain number of $(\partial)\Pi_D$ matrices. 
Let us call 
\begin{align}
\Pi_{\mathbf{3}}=\left[
\begin{array}{c|c}
\mathbf{0}_{3\times 3} & \mathbf{3} \\ \hline
\mathbf{3}^{\dagger} & \mathbf{0} 
\end{array}\right],\qquad \Pi_{\mathbf{8}\oplus \mathbf{1}}=\left[
\begin{array}{c|c}
\mathbf{8} & \mathbf{0}_{3\times 1} \\ \hline
\mathbf{0}_{1\times 3} &\mathbf{1} 
\end{array}\right].
\end{align}
Since $\mathbf{3}$ (and $\mathbf{3}^{\dagger}$) only appears off-diagonal, and we need to have a non-zero trace after the product of a certain number of $(\partial)\Pi_D$ matrices, it is clear that a non-vanishing  amplitude can only involve an even number of $(\partial)\Pi_{\mathbf{3}}$ since this combination  leads to a diagonal matrix, with any odd number of them leading to a traceless off-diagonal one. The result is independent of the basis since any of these terms is invariant under any rotation $V^{\dagger} \Pi_D V$.

We can prove this very same result in a different way, without resorting to any particular representation of the $\mathfrak{su}(4)$ algebra. One can see that there is only a $U(1)$ in $SU(4)$ outside its $SU(3)$ subgroup. Then, under $SU(3)\times U(1)$, the $\mathbf{15}$ of $SU(4)$ decomposes as
\begin{align}
\mathbf{15}=\mathbf{8}_0\oplus \mathbf{1}_0\oplus\left[\mathbf{3}_{\sqrt{\frac{2}{3}}}+\mathrm{h.c.}\right].
\end{align}
Since this extra $U(1)\subset SU(4)$ is conserved, and the $\varphi\sim \mathbf{3}_{\sqrt{\frac{2}{3}}}$ is the only irreducible representation with a non-vanishing $U(1)$ charge, DM interactions from the chiral Lagrangian will always be expressed in terms of powers of $|\varphi|^2$, involving an even number of them.

       \end{proof}
        \begin{corollary}
        \label{cor:su4}
The stable dark pions for the $SU(4)_L\times SU(4)_R\to SU(4)_V$ model can only appear an even number of times in any interaction arising from the chiral Lagrangian.
        \end{corollary}
        \begin{proof}
        Since the stable dark pions are transforming in the $\mathbf{3}$ of $SU(3)$, it follows from Propopsition~\ref{teor:su4} that one can only get an even number of them in any  interaction arising from the chiral Lagrangian.
        \end{proof}
        
        It is easy to see, this is no longer true for bigger cosets, as e.g. $SU(5)_L\times SU(5)_R\to SU(5)_V$. 
        Indeed, in this case,  the DM symmetry group is $SU(2)\times U(1)$ so that the 24 Nambu-Goldstones of $SU(5)$ can be decomposed under $SU(3)\times 
        SU(2)\times U(1)$ as~\footnote{We have taken a prefactor $ 1/(2\sqrt{15})$ out of the $U(1)$ charges to simplify the discussion. } 
\begin{align}
\mathbf{24}=(\mathbf{8},\mathbf{1})_{0}\oplus (\mathbf{1},\mathbf{1})_0 \oplus (\mathbf{1},\mathbf{3})_0\oplus \left[(\mathbf{3},\mathbf{2})_{5}+\mathrm{h.c.}\right].
\end{align}
Now, in principle, the $5^2-10=15$ stable dark pions are $\varphi\sim (\mathbf{3},\mathbf{2})_5$ and $\phi\sim (\mathbf{1},\mathbf{3})_0$. Since the $SU(3)\subset SU(5)$ is broken by the portal with the visible sector, one can write down a $SU(2)\times U(1)$ invariant  using e.g. $\sim \big(\varphi^{\dagger}\phi\, \varphi \big)_{\mathbf{1}}$.

\section{Details on the relic abundance calculation}
\label{sec:ap_relic}

The amplitude for the $2\to 2$ process  reads
\begin{align}
\mathcal{M}(\pi_D^a \pi_D^b\to \pi_D^c \pi_D^d)&=\frac{2 m_{\pi_D}^2}{3f_D^2 }\Big\{\frac{2}{n_f}\left(\delta^{ad}\delta^{bc}+\delta^{ac}\delta^{bd}+\delta^{ab}\delta^{cd}\right)+\left(d^{adm}d^{bcm}+d^{bdm}d^{acm}+d^{cdm}d^{abm}\right)\Big\}\nonumber\\
&+\frac{4}{3f_D^2}\Big\{(-4m_{\pi_D}^2+2s+t)f^{acm}f^{bdm}+(-4m_{\pi_D}^2+s+2t)f^{abm}f^{cdm}\nonumber\\
&+(s-t)f^{adm}f^{bcm}\Big\},
\end{align}
where $s,t,u$ are the usual Mandelstam variables and the indices $a,b,c,d,\ldots\in \{1,\ldots,n_f^2-1\}$. This  leads to
\begin{align}
( \sigma v)_{\rm lab}(\pi_D^a\pi_D^b\to \pi_D^c\pi_D^d)=\sigma(ab\to cd) \epsilon^{1/2}+\mathcal{O}(\epsilon^{3/2}),\quad \epsilon=\frac{s-4 m_{\pi_D}^2}{4m_{\pi_D}^2}=\bigg(\frac{v}{2}\bigg)^2\bigg(1-\left(\frac{v}{2}\right)^2\bigg),
\end{align}
with
\begin{align}
\sigma(ab\to c d)&=\frac{m_{\pi_D}^2}{72\pi f_D^4 n_f^2}\Big(n_f\big(d^{adm}d^{bcm}+d^{bdm}d^{acm}+d^{cdm}d^{abm}+8f^{adm}f^{bcm}+8f^{acm}f^{bdm}\big)\nonumber\\
&+2\big(\delta^{ad}\delta^{bc}+\delta^{ac}\delta^{bd}+\delta^{ab}\delta^{cd}\big)\Big)^2.
\end{align}
Taking into account that~\cite{Gondolo:1990dk}
\begin{align}
\langle \sigma v\rangle \approx \frac{2x^{3/2}}{\sqrt{\pi}}\int_0^{\infty}(\sigma v)_{\rm lab}\epsilon^{1/2} \exp(-x\epsilon)d\epsilon,
\end{align}
where $x=m_{\pi_D}/T$, as well as
\begin{align}
\frac{2x^{3/2}}{\sqrt{\pi}}\int_0^{\infty}\epsilon \exp(-x\epsilon)d\epsilon=\frac{2}{\sqrt{\pi}\sqrt{x}}\approx \frac{2 v}{\sqrt{3\pi}}, 
\end{align}
we obtain
\begin{align}
\langle \sigma v\rangle_{ab\to cd} &\approx \frac{m_{\pi_D}^2}{36\pi^{3/2} f_D^4 n_f^2\sqrt{x}}\Big(n_f\big(d^{adm}d^{bcm}+d^{bdm}d^{acm}+d^{cdm}d^{abm}+8f^{adm}f^{bcm}+8f^{acm}f^{bdm}\big)\nonumber\\
&+2\big(\delta^{ad}\delta^{bc}+\delta^{ac}\delta^{bd}+\delta^{ab}\delta^{cd}\big)\Big)^2.
\end{align}
For the particular case of $\mathfrak{su}(4)$, thermally averaged cross-sections involving an odd number of stable dark pions will vanish. So at the end of the day, the relevant co-annihilation cross section for this case is the one of $2$ stable dark pions going into two transient ones, 
\begin{align}
\langle \sigma v\rangle_{2_{\rm DM}\to 2_{\rm tran}}&=\frac{1}{2}\sum_{ij\in\mathcal{I}}\sum_{ab\in\mathcal{J}}\frac{\langle \sigma v\rangle_{ijab}}{N_{\rm DM}^2}\stackrel{n_f=4}{=}\frac{m_{\pi_D}^2}{\pi^{3/2} f_D^4 \sqrt{x}}\left[\frac{1171}{576}\right]\approx\frac{2m_{\pi_D}^2}{\pi^{3/2} f_D^4 \sqrt{x}}\nonumber\\
&\approx \frac{2m_{\pi_D}^2}{\sqrt{3}\pi^{3/2} f_D^4 }v=\sigma_0 v,\qquad \sigma_0=\frac{2m_{\pi_D}^2}{\sqrt{3}\pi^{3/2}f_D^4},
\end{align}
where stable dark pions have indices $\in\mathcal{I}$, and the set of indices of transient dark pions is denoted by $\mathcal{J}$. In particular, $\mathrm{card}(\mathcal{I})=N_{\rm DM}$, $\mathrm{card}(\mathcal{J})=N_{\rm tran}$ and $N_{\rm DM}+N_{\rm tran}=n_f^2-1=N_{\rm \pi}$.

On the other hand,  the thermally averaged cross-section for any 5-pion process is given by~\cite{Hochberg:2014kqa}
\begin{align}
    \langle \sigma v^2\rangle_{ijk\to lm}=\frac{m_{\pi}^5 N_d^2 T_{\{ijklm\}}^2}{96\sqrt{5}\pi^5 f_D^{10} x^2} 
\end{align}
 where $\{\ldots\}$ stands for sorting, i.e., $\{1,5,2,6,8\}=\{1,2,5,6,8\}$ and the coupling $T_{abcde}$ is given by Eq.~\eqref{eq:pionSIexp}. Again, in the $n_f=4$ case, only co-annihilation cross sections involving an even number of stable dark pions will be non-vanishing. In particular, these read
 \begin{align}
\langle \sigma v^2\rangle_{3_{\rm DM}\to 1_{\rm DM} 1_{\rm tran}}&=\frac{4}{4!}\sum_{ijkl\in \mathcal{I}}\sum_{a\in\mathcal{J}}\frac{\langle \sigma v^2\rangle_{ijkla} }{N_{\rm DM}^3}=\frac{ \sigma_5 }{N_{\rm DM}^3},\\
\langle \sigma v^2\rangle_{2_{\rm DM}1_{\rm tran} \to 2_{\rm tran}}&=\frac{3}{3! 2!}\sum_{ij\in \mathcal{I}}\sum_{abc\in\mathcal{J}}\frac{\langle \sigma v^2\rangle_{ija bc}}{N_{\rm DM}^2 N_{\rm tran}}=\frac{3\sigma_5   }{N_{\rm DM}^2 N_{\rm tran}},\\
\langle \sigma v^2\rangle_{3_{\rm tran}\to 2_{\rm DM}}&=\frac{1}{3!2!}\sum_{ij\in\mathcal{I}}\sum_{a b c \in\mathcal{J}}\frac{\langle \sigma v^2\rangle_{abc ij}}{N_{\rm tran}^3}=\frac{\sigma_5 }{N_{\rm tran}^3 },\\  
           \langle \sigma v^2\rangle_{3_{\rm tran}\to 2_{\rm tran}}&=\frac{10}{5!}\sum_{abcde\in\mathcal{J}}\frac{\langle \sigma v^2\rangle_{abcde} }{N_{\rm tran}^3}=\frac{5\sigma_5 }{2 N_{\rm tran}^3},\\
              \langle \sigma v^2\rangle_{1_{\rm DM}2_{\rm tran}\to 1_{\rm DM}1_{\rm tran}}&=\frac{9}{3!2!}\sum_{abc\in\mathcal{J}}\sum_{ij\in\mathcal{I}}\frac{\langle \sigma v^2\rangle_{abcij} }{N_{\rm DM}N_{\rm tran}^2}           =\frac{9\sigma_5}{N_{\rm DM}N_{\rm tran}^2},\\
           \langle \sigma v^2\rangle_{2_{\rm DM}1_{\rm tran}\to 2_{\rm DM}}&=\frac{6}{4!}\sum_{a\in\mathcal{J}}\sum_{ijkl\in\mathcal{I}}\frac{\langle \sigma v^2\rangle_{aijkl} }{N_{\rm DM}^2 N_{\rm tran}}=\frac{3 \sigma_5}{2 N_{\rm DM}^2 N_{\rm tran}}.
           \end{align}

     where 
     \begin{align}
\sigma_{5}=\frac{25 m_{\pi}^5 N_d^2}{32\sqrt{5}\pi^5 f_D^{10} x^2}.
     \end{align}
 
The Boltzmann equations become
  \begin{align}
            \frac{1}{a^3}\frac{d}{dt}\big(n_{\rm DM} a^3\big)=&-\langle \sigma v\rangle_{2_{\rm DM}\to 2_{\rm tran}}\left[n_{\rm DM}^2-(n_{\rm DM}^2)_{\rm eq} \right]\nonumber\\
            &      -2\langle \sigma v^2 \rangle_{3_{\rm DM}\to 1_{\rm DM} 1_{\rm tran}}\Big[n_{\rm DM}^3-\Bigg(\frac{n_{\rm DM}^2}{n_{\rm tran}} \Bigg)_{\rm eq} n_{\rm DM} n_{\rm tran}\Big]\nonumber   \\
      &-2\langle \sigma v^2\rangle_{2_{\rm DM} 1_{\rm tran}\to 2_{\rm tran}} \Big[n_{\rm DM}^2n_{\rm tran}-\Bigg(\frac{n_{\rm DM}^2}{n_{\rm tran}} \Bigg)_{\rm eq}  n_{\rm tran}^2\Big]\nonumber\\
      &+2\langle \sigma v^2\rangle_{3_{\rm tran}\to 2_{\rm DM}} \Big[n_{\rm tran}^3-\Bigg(\frac{n_{\rm tran}^3}{n_{\rm DM}^2} \Bigg)_{\rm eq}  n_{\rm DM}^2\Big],\\
            \frac{1}{a^3}\frac{d}{dt}\big(n_{\rm tran} a^3\big)=&+\langle \sigma v\rangle_{2_{\rm DM}\to 2_{\rm tran}}\left[n_{\rm DM}^2-(n_{\rm DM}^2)_{\rm eq} \right]\nonumber\\
   &   -\langle \sigma v^2 \rangle_{3_{\rm tran}\to 2_{\rm tran}}\Big[n_{\rm tran}^3-(n_{\rm tran})_{\rm eq} n_{\rm tran}^2\Big]\nonumber   \\
      &-3\langle \sigma v^2\rangle_{3_{\rm tran} \to 2_{\rm DM}} \Big[n_{\rm tran}^3-\Bigg(\frac{n_{\rm tran}^3}{n_{\rm DM}^2} \Bigg)_{\rm eq}  n_{\rm DM}^2\Big]\nonumber\\
      &-\langle \sigma v^2\rangle_{ 1_{\rm DM}2_{\rm tran}\to 1_{\rm DM}1_{\rm tran}} \Big[n_{\rm tran}^2 n_{\rm DM}-(n_{\rm tran})_{\rm eq} n_{\rm tran} n_{\rm DM}\Big]\nonumber\\
      &+\langle \sigma v^2\rangle_{2_{\rm DM}1_{\rm tran} \to 2_{\rm tran}} \Big[n_{\rm tran}n_{\rm DM}^2 -\Bigg(\frac{n_{\rm DM}}{n_{\rm tran}}\Bigg)_{\rm eq} n_{\rm tran}^2\Big]\nonumber\\
       & -\langle \sigma v^2 \rangle_{2_{\rm DM}1_{\rm tran}\to 2_{\rm DM}}\Big[n_{\rm DM}^2n_{\rm tran}-(n_{\rm tran})_{\rm eq} n_{\rm DM}^2\Big]\nonumber   \\
    &  +\langle \sigma v^2\rangle_{3_{\rm DM} \to 1_{\rm DM}1_{\rm tran}} \Big[n_{\rm DM}^3-\Bigg(\frac{n_{\rm DM}^2}{n_{\rm tran}} \Bigg)_{\rm eq}  n_{\rm DM}n_{\rm tran}\Big]-\Gamma(\pi_{\rm tran}) n_{\rm tran}.
           \end{align}

\section{Details on the recast of LHC searches}
\label{app:selection}
\begin{table}[!th]
\centering
\resizebox{1.\textwidth}{!}{% 
  \begin{tabular}{|c|c|c|}
        \hline
        \textbf{four jets} & t\textbf{wo jets plus MET} & \textbf{semi-visible} \\
        \hline 
        $N_{j} \geq 4$ with $p_T > 80$~GeV, $|\eta|< 2.5$ & $N_{jet}\geq 2$ with $p_T > 30$~GeV, $|\eta|< 2.4$ & $N_{jet}\geq 2$ with $p_T > 30$~GeV, $|\eta|< 2.8$\\
        \hline
        $H_T < 1050$~GeV or $\geq 1$ jet with $p_T > 550$~GeV & $H_T > 300$~GeV in $|\eta| < 2.4 $ and $H_T^{miss} > H_T$ & $H_T > 600$~GeV and leading jet $p_T > 250$~GeV \\
        \hline
         & $H_T^{miss} > 300$~GeV in $|\eta| < 5$ & $H_T^{miss} > 600$~GeV\\
         \hline
         & no isolated leptons with $p_T > 10$~GeV & no leptons with $p_T > 7$~GeV \\
         \hline
         & & $\leq 1$ b-jet \\
         \hline
         & $\Delta\phi(\vec{H}_T^{miss},j_{1,2}) > 0.5$ and $p_T$ jets $\Delta\phi(\vec{H}_T^{miss},j_{>2}) > 0.3$ & as least one jet with $\Delta\phi(\vec{H}_T^{miss},j) = 2$\\
         \hline
         $\Delta R_{1,2} < 2$, $\Delta\eta < 1.1$ and asymmetry $< 0.1$ & & \\
         \hline
    \end{tabular}%
}
\caption{Signal selection criteria for four jets~\cite{CMS:2022usq}, two jets plus missing energy~\cite{CMS:2019zmd} and semi-visible jets searches~\cite{ATLAS:2022aaw}. }
\label{table:selection} 
\end{table}

The selection criteria used for the recast of the non-resonant dijet, jets plus missing energy and semi-visible jets searches are given in table~\ref{table:selection}. We discuss them here in more detail.

In the non-resonant dijet search events are selected if they have $N_{jets} \geq 4$ with $p_T > 80$~GeV, $|\eta| < 2.5$, the scalar $p_T$ sum $H_T > 1050$~GeV or at least one jet with $p_T > 550$. In addition to these selection criteria the four hardest jets are combined to two dijets based on the combination that minimizes $\Delta R = |\left(\Delta R_1-0.8\right)|+|\left(\Delta R_2-0.8\right)|$, where $\Delta R_i$ is the $\eta-\phi$ separation between two jets. Based on the minimizing combination signal events are also required to have $\Delta R_{1,2}<2$, $\Delta\eta = |\eta_1-\eta_2| < 1.1$ and an asymmetry $= \frac{|m_1-m_2|}{m_1+m_2} < 0.1$. The small asymmetry requirement is based on the assumption that two close jets originate from the same resonance, so that the two dijets are close in mass. This is not the case here, however, because, if combined correctly, one jet in a dijet contains stable dark pions. Thus, this requirement reduces the efficiency of this search for our scenario drastically as can be seen by the purple lines in Fig.~\ref{fig:efficiency}. For large dark pion lifetimes the $H_T$ and $p_T$ requirements are no longer fulfilled and the efficiency decreases even further.

The two jets plus MET search at CMS the requires at least two jets $N_{jet}\geq 2$ with $|\eta| < 2.4$ and $p_T > 30$~GeV. Additional requirements demand that the scalar $p_T$ sum $H_T > 300$~GeV for jets with $|\eta| < 2.4$ and the absolute value of the negative $p_T$ vector sum $H_T^{miss} > 300$~GeV for jets with $|\eta|<5$. To avoid mismeasurement events need to fulfill $H_T^{miss} < H_T$. Events with isolated electrons (muons) with $I < 0.1 (0.2)$, where $I$ is the scalar $p_T$ sum of hadrons and photons in a cone around the leptons over the letpon $p_T$ are excluded. Finally, an azimuthal angle between $\vec{H}_T^{miss}$ and two highest $p_T$ jets $\Delta\phi(\vec{H}_T^{miss},j_{1,2}) > 0.5$ and for lower $p_T$ jets $\Delta\phi(\vec{H}_T^{miss},j_{>2}) > 0.3$ is required for signal events.  Based on these selection criteria the search remains efficient even at small lifetimes, as stable dark pions produced in jets originating from dark quarks are enough to fulfill the CMS selection criteria. At around $10$~mm the efficiency increases, since more and more dark pions are stable on collider scales and the criteria on missing energy and $\Delta\Phi$ are fulfilled for more events. \\
Finally, for the semi-visible search selection criteria  on the number of jets, $N_{jet}\geq 2$ with $p_T > 30$~GeV and $|\eta|< 2.8$, the scalar $p_T$ sum $H_T > 600$~GeV, leading jet momentum $p_T > 250$~GeV and missing energy $H_T^{miss} > 600$~GeV are used. Moreover, no muons with $p_T > 7$~GeV in $|\eta|< 2.5$ and no electrons with $p_T > 7$~GeV in $|\eta|< 2.5$ and $1.37< |\eta| < 1.52$ are allowed for signal events. Events with more than one b-jet are rejected and at least one jet with $\Delta\phi(\vec{H}_T^{miss},j) = 2$ is required to pass the signal selection. From Fig.~\ref{fig:efficiency} it can be seen that this search has the smallest efficiencies. Because dark quarks are mainly produced via t-channel mediator exchange they lead to low momentum dark jets, so that the $H_T^{miss}$ requirement rules out most events. As expected the search becomes more efficient for larger dark pion lifetimes, since, similarly to the jets plus MET search, the MET requirements becomes more often fulfilled the more dark pions become stable on collider scales. An interesting behaviour can be seen in the case for couplings to down quarks: Dark pions with lower masses can be searched for more efficiently. This can be explained by the veto on b-jets: As soon as the b-channel is kinematically allowed most dark pions decay to b quarks leading to the decrease of the efficiency. The limits are given for six values of $R_{inv}$, which in our model increases as the lifetime of the transient dark pions increases. In practice we determine $R_{inv}$ as the fraction of produced dark pions that are either stable or decay outside of the hadronic calorimeter based on simulated events. We find $R_{inv}$ values ranging from $\sim0.3$ for prompt decays to $\sim0.95$, when transient dark pions are long-lived, and apply the appropriate limits.

\section{Additional figures}
\label{app:xsec}
Fig.~\ref{fig:xseclimit} shows the upper limits on the cross section for the non-resonant dijet search~\cite{CMS:2022usq} and jets plus MET~\cite{CMS:2019zmd} (left) and the cross section limits of the semi-visible search~\cite{ATLAS:2022aaw} for $R_{inv}=0.2$, 0.4, 0.6 and 0.8 (right).
\begin{figure}[!h]
    \centering
    \includegraphics[width=0.47\linewidth]{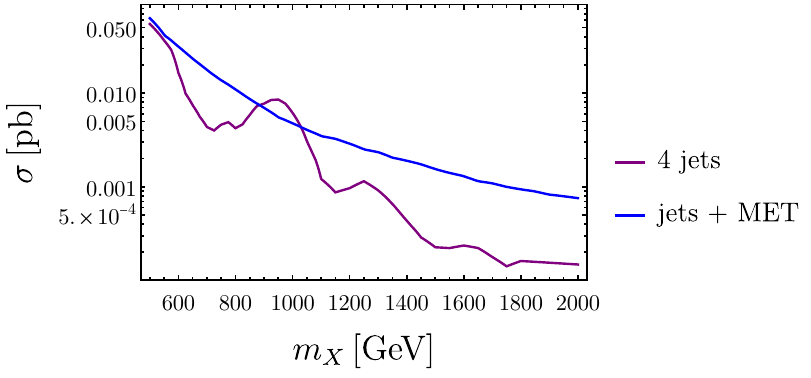}
    \quad
    \includegraphics[width=0.46\linewidth]{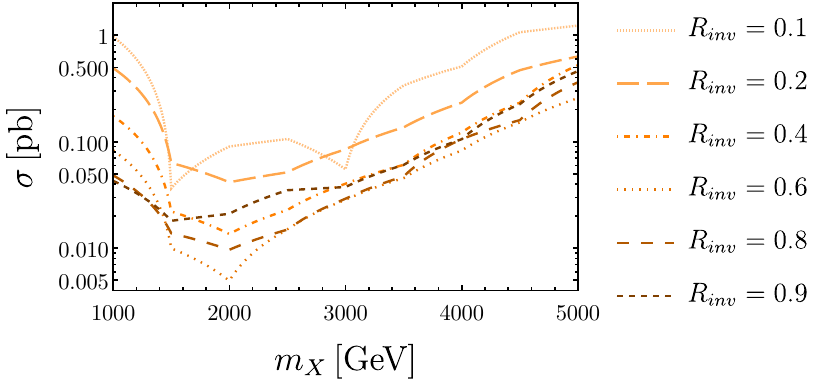}
    \caption{Left: Observed upper limit on cross section for the non-resonant dijet search (purple) in~\cite{CMS:2022usq} and the CMS two jets plus missing energy search (blue) from~\cite{CMS:2019zmd} for the lowest neutralino mass of $\sim 6$~GeV. 
    Right: Observed upper limit on cross section for the ATLAS semi-visible search~\cite{ATLAS:2022aaw} for different values of $R_{inv}$.}
    \label{fig:xseclimit}
\end{figure}  

In Fig.~\ref{fig:mxlimit01} we show the upper limits on the mediator mass as a function of the dark pion lifetime, same as in Fig.~\ref{fig:mX_limit}, but for $\kappa=0.1$. Due to the $\kappa^4$ suppression of the t-channel mediator exchange production cross section the fiducial cross section for the semi-visible search is always below the experimental limit.
\begin{figure}[!h]
    \centering
    \includegraphics[width=0.48\textwidth]{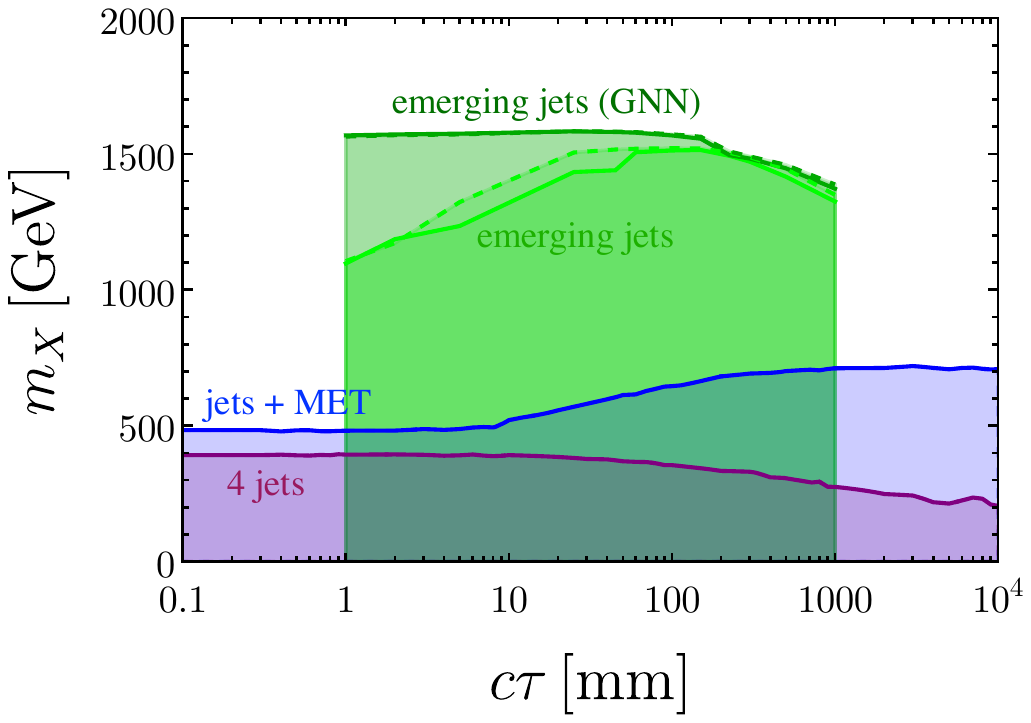}\quad
\includegraphics[width=0.48\textwidth]{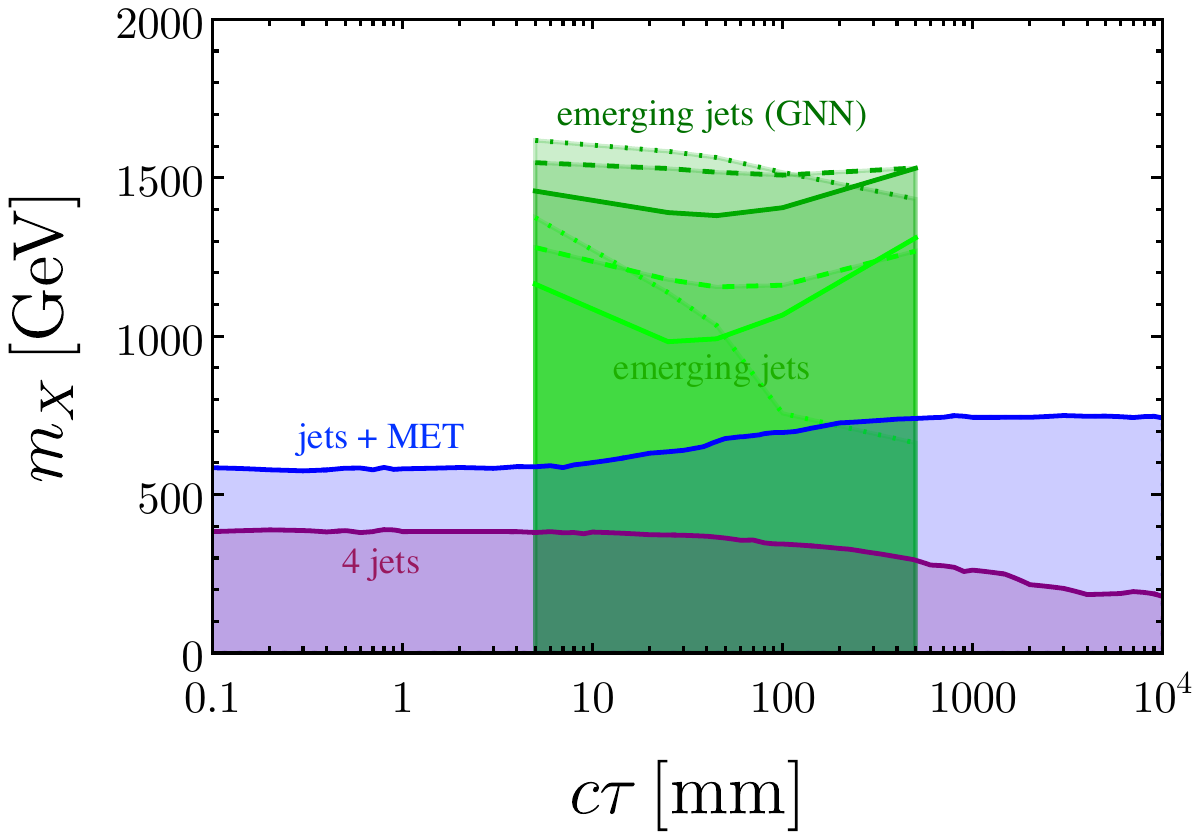}
    \caption{Upper limit on $m_X$ from the emerging jet search (solid, dashed and dotted green and dark green show the limits from the model agnostic and GNN search with the full run data set for $m_{\pi_D} = 10,\, 20,\, 6$ GeV), MET search (blue) and four jets search (purple) for couplings $\kappa = 0.1$ and $m_{\pi_D} = 1 $\,GeV to up-type (left) and down-type quarks (right).
    }
    \label{fig:mxlimit01}
\end{figure}

Finally, Figs.~\ref{fig:mx_mpi15} and~\ref{fig:mx_mpi15_down} show the combined collider, direct detection and flavour constrains in the $\mpid-m_X$ frame for $\kappa = 1$ and 0.1, respectively, for couplings to up-type (left) and couplings to down-type quarks (right). In both figures $f_D = 15\mpid$ was chosen.
\begin{figure}[!h]
    \centering
    \includegraphics[width=0.47\linewidth]{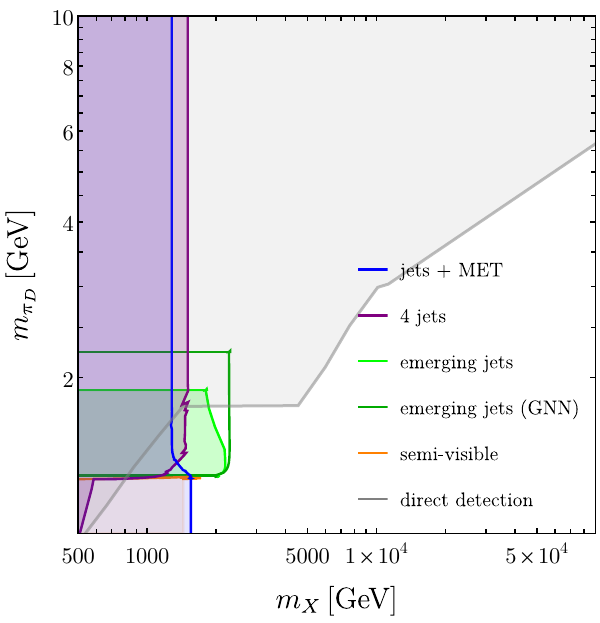}\quad
    \includegraphics[width=0.47\linewidth]{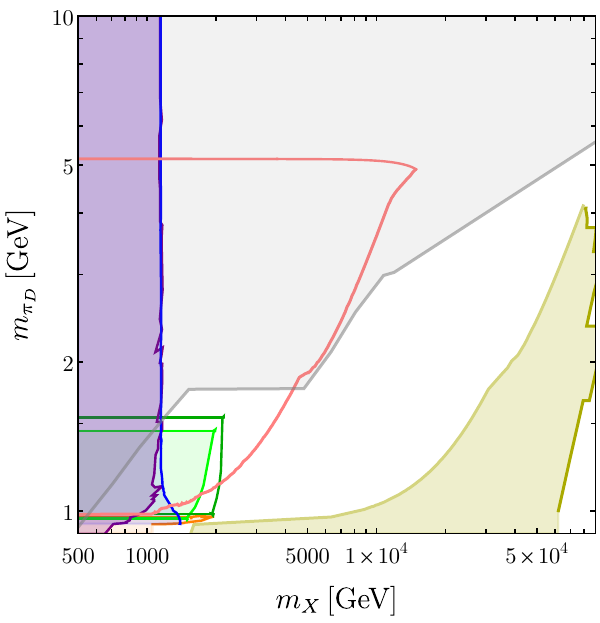}
    \caption{Same as Fig.~\ref{fig:mx_mpi1}, but for $f_D=15\mpid$.}
    \label{fig:mx_mpi15}
\end{figure}

\begin{figure}
    \centering
    \includegraphics[width=0.47\linewidth]{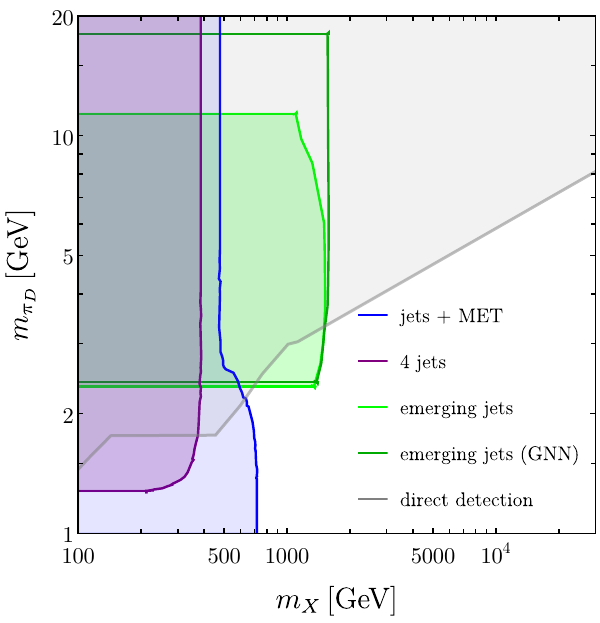}\quad
    \includegraphics[width=0.47\linewidth]{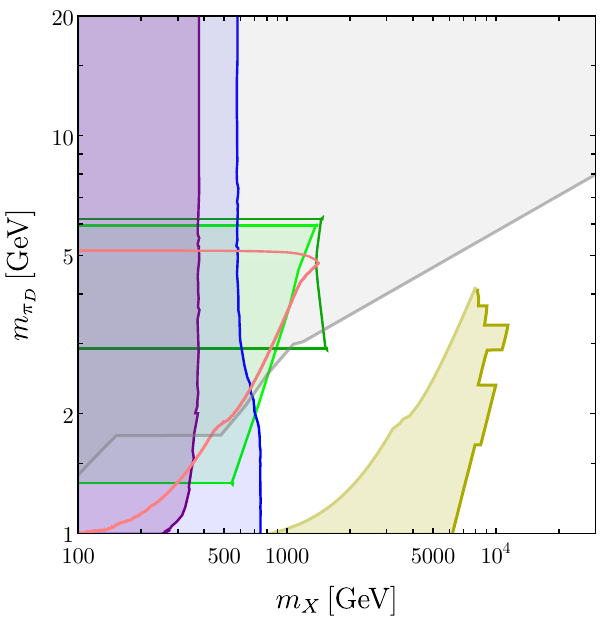}
    \caption{Same as Fig.~\ref{fig:mx_mpi1_down} but for $f_D=15\mpid$.}
    \label{fig:mx_mpi15_down}
\end{figure}
\clearpage
        \bibliography{refs}{}
\bibliographystyle{JHEP}
\end{document}